\pdfoutput=1
\RequirePackage[T1]{fontenc}
\RequirePackage{fix-cm}
\RequirePackage{mlmodern}

\documentclass[11pt,a4paper,english]{amsart}
\usepackage[svgnames]{xcolor}
\usepackage{babel}
\usepackage[latin1]{inputenc}
\usepackage{amsmath}
\usepackage{amsthm}
\usepackage{amsfonts}
\usepackage{indentfirst}
\usepackage{graphicx}
\usepackage{amssymb}
\usepackage[mathscr]{eucal}
\usepackage{tikz}
\usepackage{caption}
\usepackage{subcaption} 

\usepackage{booktabs}
\usepackage{array}
\usepackage[normalem]{ulem}

\usepackage{enumerate}
\usepackage{xspace}
\usepackage{xfrac}
\usepackage{faktor}
\usepackage{relsize}
\usepackage{float}
\usepackage{multirow}
\usepackage{ifsym}
\usepackage{changepage}

\usepackage[colorlinks=true,citecolor=black,linkcolor=black,urlcolor=blue]{hyperref}
\usepackage{fullpage,epsf}
\usepackage{amscd}

\usepackage{rotating}

\usepackage{adjustbox}

\usepackage{pdflscape}

\newtheorem{teo}{Theorem}

\newtheorem{pro}[teo]{Proposition}
\newtheorem{cor}[teo]{Corollary}
\newtheorem{lem}[teo]{Lemma}

\theoremstyle{definition}
\newtheorem{rem}[teo]{Remark}
\newtheorem{de}[teo]{Definition}

\newcommand{\F}{{\mathbb{F}}}

\newcommand{\Z}{{\mathbb Z}}

\newcommand{\ket}[1]{|#1 \rangle}
\newcommand{\bra}[1]{\langle #1 |}

\DeclareMathOperator{\dis}{d}
\DeclareMathOperator{\ev}{ev}

\newcommand{\Co}{\mathcal{C}}

\newcommand{\Su}{\mathcal{S}}

\title[Quantum $(r,\delta)$-Locally Recoverable Codes]{Quantum $(r,\delta)$-Locally Recoverable BCH and Homothetic-BCH Codes}
\author[C. Galindo, F. Hernando and R. Matsumoto]{Carlos Galindo, Fernando Hernando and Ryutaroh Matsumoto}
\curraddr{\texttt{Carlos Galindo and Fernando Hernando:} Instituto
Universitario de Matem\'aticas y Aplicaciones de Castell\'on and
Departamento de Matem\'aticas, Universitat Jaume I, Campus de Riu
Sec. 12071 Castell\'{o} (Spain)\\
\texttt{Ryutaroh Matsumoto:} Department of Information and Communications Engineering, Institute of Science Tokyo, Japan.}

\email{{\rm Galindo: galindo@uji.es; {\rm Hernando:} carrillf@uji.es; {\rm Matsumoto:} ryutaroh@ict.e.titech.ac.jp}}

\urladdr{{\rm Galindo: 0000-0002-3908-4462; {\rm Hernando:} 0000-0002-9758-2152; {\rm Matsumoto:} 0000-0002-5085-8879}\color{black}}

\thanks{The first two authors were partially funded by MCICIU/AEI/10.13039/501100011033 and by ``ERDF, UE'' (grant PID2022-138906NB-C22). The third author was partially supported by ``Japan Society of the Promotion of the Science'' (grant 23K10980). }

\date{26 July 2026 (version 2); 30 January 2026 (version 1)}

\begin{document}

\begin{abstract}
Quantum $(r,\delta)$-locally recoverable codes ($(r,\delta)$-LRCs) are the quantum version of classical $(r,\delta)$-LRCs designed to recover multiple failures in large-scale distributed and cloud storage systems. A quantum $(r,\delta)$-LRC, $Q(C)$, can be constructed from an $(r,\delta)$-LRC, $C$, which is Euclidean or Hermitian dual-containing.

This article is devoted to studying how to get quantum $(r,\delta)$-LRCs from BCH and homothetic-BCH codes. As a consequence, we give pure quantum $(r,\delta)$-LRCs which are optimal for the Singleton-like bound.
\end{abstract}

\maketitle

\section{Introduction}\label{se:uno}
Established between 1959 and 1960, Bose-Chaudhuri-Hocquenghem (BCH) codes represent a fundamental class of cyclic codes within classical coding theory, widely featured in standard information theory literature. We regard BCH codes  as subfield-subcodes of suitable $\{1\}$-affine variety codes, see Definition \ref{AA} and \cite{Cas, HBCH}. $\{1\}$-affine variety codes are evaluation codes based on univariate polynomials \cite{QINP} and our perspective allows for the generalization of BCH codes into subfield-subcodes of $J$-affine variety codes, which utilize multivariate polynomial evaluation \cite{galindo-hernando, QINP, QINP2}. Evaluation codes are a large class of codes that extend and facilitate the study of important families of codes, like algebraic geometry codes \cite{H-VL-P} or the aforementioned $J$-affine variety codes. Homothetic-BCH (H-BCH) codes were introduced in \cite{HBCH}, see Definition \ref{definit}. They offer block lengths unattainable by $\{1\}$-affine variety codes and a BCH code $\mathcal{B}$ can be derived by appropriately puncturing a homothetic-BCH code $\mathcal{C}$.

In this article, we study how to get {\it quantum} $(r,\delta)$-{\it locally recoverable codes} ($(r,\delta)$-LRCs) with the help of BCH and H-BCH codes. From our study we derive some families of {\it optimal pure quantum} $(r,\delta)$-LRCs. In the next paragraphs, we recall some very basic information about the codes we are interested in. It will be expanded in Section \ref{preli}.

The literature provides significant evidence for quantum supremacy \cite{Aru, BallP, Zhong, LiuY}, highlighting the transformative potential of quantum information processing across various fields. To utilize quantum hardware effectively, protecting information from decoherence and environmental noise is vital. Quantum error-correcting codes (QECCs) are essential for this purpose, as they enable reliable quantum correction \cite{ShorS, 95kkk} while adhering to the fundamental constraints of the no-cloning theorem \cite{Woot}. While foundational research in QECCs focused exclusively on binary systems \cite{20kkk, Gottesman, Calder2, Calderbank}, subsequent studies have extensively explored $q$-ary codes, where $q$ represents a prime power. These have been developed through a broad range of methodologies, as documented in works such as \cite{BE, AK, Ketkar, Aly, XingC, Lag2, gahe, Anderson}.

Let $\mathbb{F}_q$ denote the finite field of size $q$. A QECC of length $n$ is defined as a $K$-dimensional linear subspace within the complex Hilbert space $\mathbb{C}^{q^n}$. We are interested in the so-called {\it stabilizer quantum codes} because the existence of a stabilizer code is tied to the existence of an additive classical code $C \subseteq \mathbb{F}_q^{2n}$ that is self-orthogonal under the trace-symplectic form, allowing researchers to utilize classical results. By employing more accessible inner products, such as the Hermitian product on $\mathbb{F}_{q^2}^{n}$ or the Euclidean product on $\mathbb{F}_{q}^{n}$, one can derive stabilizer quantum codes from self-orthogonal classical structures, see Theorem \ref{resto}.  Aly et al. \cite{Aly} established distance bounds for narrow-sense BCH codes that ensure Euclidean or Hermitian dual containment. Building on this, several works \cite{LMFL2013, KZT2013, Lag2, HZC2015, XLGM2016, LLLM2017, YZKL2017, QZ2017, LLLG2017, ZSL2018} have introduced diverse  families of QECCs based on classical BCH, negacyclic, or constacyclic codes, providing enhanced parameters for specific block lengths. Several recent studies have focused on Hermitian dual-containing BCH codes as a source for quantum stabilizer codes. Among these are \cite{SYW2019, LS2021, ZZ2021, L2022}. Returning to the case of the above mentioned BCH codes $\mathcal{B}$ and H-BCH codes $\mathcal{C}$, \cite[Theorem 2.6]{HBCH} proves that the Hermitian self-orthogonality of $\mathcal{B}$ implies that of $\mathcal{C}$. The dimension and minimum distance of H-BCH codes can be bounded similarly to standard BCH codes. Thus, \cite[Theorem A]{HBCH} uses these properties, and a significant refinement given in the same theorem of the bounds presented in \cite[Theorem 13]{Aly}, for obtaining an extensive new class of Hermitian self-orthogonal codes suitable for quantum applications.

The preservation of data across large-scale distributed and cloud storage systems is a critical concern for modern enterprises. In \cite{GHSY2012}, the application of classical error-correcting codes was proposed  to address the recovery of data from a failed node using information from surviving nodes; a code $C$ of this type is characterized as having $r$-locality if any specific coordinate of a codeword in $C$ can be reconstructed by accessing a maximum of $r$ other symbols. Such codes are designated as $r$-locally recoverable codes ($r$-LRCs). Some advancements in this field are documented in \cite{TPD2016, BTV2017, LMC2018, Mi2018, LXY2019, J2019, LMT2020, SVV2021, edgar}.

While standard $r$-locality accounts for single node failures, the reality of simultaneous multiple device outages necessitates a more robust approach. To mitigate this, $(r, \delta)$-{\it locally recoverable codes} were proposed in \cite{PKLK2012}. Fixing a coordinate $c$, these codes can rectify up to $\delta-1$ erasures within a recovery set containing $c$ and at most $r + \delta - 2$ additional coordinates. Parameters of such codes are governed by a Singleton-like bound; those codes attaining it are classified as {\it optimal} $(r, \delta)$-{\it LRCs}. Recent research has focused heavily on constructing these optimal structures \cite{SDYL2014, CXHF2018, LMX2019, LXY2019, SZW2019, J2019, CFXF2019, Z2020, FF2020, QZF2021, KWG2021, Cai, GFMC}.

{\it Quantum locally recoverable codes} (quantum LRCs) represent the extension of the above principles to quantum information. Initial developments by Golowich et al. \cite{Golowich, Golowich2} defined quantum $r$-LRCs, which, analogous to their classical counterparts, are QECCs capable of correcting an erasure of a single input qudit. Recently, and more broadly, the concept of {\it quantum} $(r, \delta)$-{\it LRC}, $Q \subseteq \mathbb{C}^{q^n}$, has been coined \cite{QLRC}. A code of this type must satisfy that, for every index $i \in \{1, \dots, n\}$, there exists a recovery set $J$ containing $i$ with cardinality less than or equal to $r + \delta - 1$ such that $Q$ is $(I, J)$-locally recoverable for every subset $I \subseteq J$ of size $\delta - 1$, see Section \ref{sect25}.

Let $C$ be a Hermitian or Euclidean dual-containing classical linear code giving rise to a quantum stabilizer code $Q'(C)$. Theorem 28 in \cite{QLRC} establishes that, when the condition $\delta \leq \dis(C^{\perp_h})$ for the Hermitian case (or $\delta \leq \dis(C^{\perp_e})$ for the Euclidean case) is satisfied, $\dis$ means minimum distance, $Q'(C)$ is a quantum $(r, \delta)$-LRC if and only if $C$ is a classical $(r,\delta)$-LRC; note that $\perp_h$ (respectively, $\perp_e$) stands for Hermitian (respectively, Euclidean) dual. Additionally, if $C$ is a Hermitian or Euclidean dual-containing $(r, \delta)$-LRC of length $n$ and dimension $\frac{n+k}{2}$, then $Q'(C)$ is a quantum stabilizer $(r,\delta)$-LRC with parameters $[[n,k, \geq \dis(C)]]_q$, which satisfies the following Singleton-like bound:
\begin{equation}\label{AAAA} k + 2 \dis(C) + 2\left(\left\lceil\frac{n+k}{2r}\right\rceil-1\right)(\delta-1) \leq n+2.\end{equation}
{\it Pure quantum codes} $Q'(C)$ attaining equality in (\ref{AAAA}) are named {\it optimal}.

We are mainly interested in subfield-subcodes over $\mathbb{F}_q$ of codes supported in larger fields $\mathbb{F}_{q^s}$. Let $n$ and $N$ be two positive integers such that $n$ divides $N$ which divides $q^s -1$. Denote by $P^{\mathfrak{A}}$ the set of $N$-roots of unity (under a suitable ordering, see Subsection \ref{Subs31}) and by $R$
a subset of $P^{\mathfrak{A}}$,  which depends on a positive integer $\lambda \leq \frac{N}{n}$, defined as the set $P$ in (\ref{eq:pointsetHBCH}) in Subsection \ref{22} but without any restriction on $\lambda$; when $\lambda = \frac{N}{n}$, $R= P^{\mathfrak{A}}$. In this article we consider evaluation codes $\mathcal{H}_\Delta^R$ which evaluate univariate polynomials in the $\mathbb{F}_{q^s}$-linear space generated by $X^e$, $e \in \Delta$, at the points in $R$. Subfield-subcodes over $\mathbb{F}_{q}$ of these codes are denoted $\mathcal{S}_\Delta^{R,q}$. When $\lambda = \frac{N}{n}$ and $\Delta$ is a suitable set, we call these subfield-subcodes BCH codes of length $N$, see Definition \ref{AA}; for any $\Delta$ and each $\lambda$ such that $\lambda n$ divides $N$, one gets a subfield-subcode of a $\{1\}$-affine variety code of length $\lambda n$. In case $\lambda n$ does not divide $N$, we set $P$ instead of $R$ and the code $\mathcal{S}_\Delta^{P,q}$ is named H-BCH.

Our goal is to get {\it pure quantum} $(r,\delta)$-{\it LRCs}, as good as possible, from Euclidean and Hermitian dual codes of  codes $\mathcal{S}_\Delta^{R,\mathfrak{q}}$.  Note that, for using Hermitian duality, we need $s$ to be even, which we write $s=2\varsigma$, and, to unify notation, we set $\mathfrak{q}=q^2$ in this case and $\mathfrak{q}=q$ in the Euclidean one.

To achieve our goal, we need to provide sets $\Delta$ for which one can regard the Euclidean (respectively, Hermitian) dual $\left( \mathcal{S}_\Delta^{R,q} \right)^{\perp_e}$ (respectively, $\left( \mathcal{S}_\Delta^{R,q^2} \right)^{\perp_h}$) as an $(r,\delta)$-LRC. This is done in Theorem \ref{eval} when evaluating at the set $P^{\mathfrak{A}}$ and, supported in this result, for H-BCH codes in Theorem \ref{mats-2}.

We look for {\it quantum} $(r,\delta)$-LRCs and, thus, we desire that our supporting classical codes be not only $(r,\delta)$-LRC but also Euclidean or Hermitian dual-containing. Theorem \ref{qlrcbch} restricts our previous conditions on $\Delta$ to reach the objective.

As mentioned, there is a definition of optimal quantum $(r,\delta)$-LRCs when the quantum codes are pure. By considering specific families of sets $\Delta$ as introduced in the before results, in Theorems \ref{4-5-japo} and \ref{jap-s-3}, and Corollary  \ref{BCHER}, we show parameters of optimal pure quantum $(r,\delta)$-LRCs. Even more, evaluating at multivariate polynomials, optimal codes with larger lengths are obtained, see Theorem \ref{teomonomial}. Note that our codes have different parameters of the specific families of optimal codes supplied in the literature, \cite{QMP-1, QMP-3, QMP-2, QMP-4, QMP-5}. In particular we provide binary and ternary optimal pure quantum $(r,\delta)$-LRCs with unbounded length. Table \ref{tab:quantum-lrcs1} in page \pageref{tab:quantum-lrcs1} lists parameters of quantum $(r,\delta)$-LRCs brought by our proposed constructions. Moreover, the sets of realizable parameters of QLRCs by the previous researches are summarized in Tables \ref{tab:quantum-lrcs2} to  \ref{tab:quantum-lrcs6}. From the tables, one can see that many
new parameters can be realized, particularly over small quantum alphabets.
Since quantum memory in the current quantum computers are mostly binary or ternary, our new constructions significantly extend available parameters of quantum $(r,\delta)$-LRCs applicable to them.

We conclude with a brief outline of the paper. Section \ref{preli} reviews important concepts for our study related to classical error-correcting codes ($\{1\}$-affine variety, BCH and H-BCH codes, and local recoverability) and quantum codes (stabilizer codes, relation between quantum and classical LRCs, and optimality of pure quantum LRCs). The main results of the paper appear in Section \ref{LA3}. Subsection \ref{Subs31} introduces sets $\Delta$ for which we are able to determine parameters and $(r,\delta)$-locality of the Euclidean and Hermitian dual codes of the subfield-subcodes $\mathcal{S}_\Delta^{R,\mathfrak{q}}$. Subsection \ref{sect32} states conditions for the above dual codes to be Euclidean or Hermitian dual-containing giving rise to quantum codes. To finish, Section \ref{LA4}, supported in the results in Section \ref{LA3}, provides several families of new optimal pure quantum $(r,\delta)$-LRCs.

\section{Preliminaries}
\label{preli}
Let $q$ be a prime power. In this paper, as supporting linear codes we will consider $q$-ary subfield-subcodes of certain linear codes over finite fields $\mathbb{F}_{q^{s}}$, $s \geq 2$. We will also use $q^2$-ary subfield-subcodes in case $s= 2 \varsigma$, $\varsigma$ being a positive integer. Let us start with introducing $\{1\}$-affine variety codes and BCH codes regarded as evaluation codes.
\subsection{\texorpdfstring{$\{1\}$}{af}-affine variety and BCH codes}
\label{BCH-1}

Let $N$ be a divisor of $q^{s} - 1 $ and denote by $a$ a primitive $N$-th root of unity.  Let $\langle X^{N} -1 \rangle$ be the ideal of the polynomial ring $\F_{q^{s}}[X]$  generated by $X^{N} -1$, and set $U(N) = \{1,a, \ldots, a^{N-1}\}$ as its set of zeroes. Consider the following evaluation map:
$$\ev^{U(N)}: \faktor{\F_{q^{s}}[X]}{\langle X^{N} -1 \rangle} \to \F_{q^{s}}^{N} \textrm{, } \quad \ev^{U(N)}(h)=\left(h(1), h(a),\dots,h(a^{N-1})\right),$$
$h$ being a class in $\faktor{\F_{q^s}[X]}{\langle X^{N} -1 \rangle}$, which is usually specified  by the representative of minimum degree in its residue class.

Consider the ring $\Z / N\Z$ of congruences modulo $N$ and its set of representatives $\mathbb{Z}_N := \{0, 1, \ldots, N-1\}$. For any subset $\Delta \subseteq \mathbb{Z}_N$, the $\{1\}$-{\it affine variety code} of length $N$ given by $\Delta$, $\Co_{\Delta}^{N}$, is the linear code defined as follows:
    $$\Co_{\Delta}^{N}:=\langle \ev^{U(N)}(X^e) \; : \; e \in \Delta \rangle \subseteq \F_{q^{s}}^{N},$$
here $\langle \cdot\rangle$ means generation as an $\F_{q^{s}}$-linear space. Let $\Lambda^q_e$ stand for the cyclotomic coset in $\mathbb{Z}/N \mathbb{Z}$ given by $e \in \mathbb{Z}_N$ with respect to $q$. That is, $\Lambda^q_e=\{q^{i}e \; : \; i\geq 0\}\subseteq \mathbb{Z}_N$. For simplicity, our element $e$ representing $\Lambda^q_e$ is the minimum element in $\Lambda^q_e$ with respect to the usual order. Denote by
$$\mathcal{A}_q(N)=\{\alpha_0(N)=0 < \alpha_1(N)< \cdots <\alpha_{\nu}(N)\}\subseteq \mathbb{Z}_N$$
the set of representatives for the above cyclotomic cosets.

For each pair $(t,t')$, $0 \leq t \leq t' < \nu$, set $\Delta_q^{N} (t,t') = \cup_{\ell =t}^{t'} \Lambda^q_{\alpha_\ell}$, where we have deleted $(N)$ for simplicity. Then (see, \cite{Bier, Cas}) a $(t,t')$-BCH code can be defined as follows:
\begin{de}
\label{AA}
The $q$-ary BCH code of length $N$ given by the pair $(t,t')$, $\mathrm{BCH}_q^N(t,t')$, is defined as the subfield-subcode $\mathcal{C}^{N}_{\Delta_q^{N} (t,t')} \cap \mathbb{F}_{q}^N$.
\end{de}

Denote by $C^{\perp_e}$ the Euclidean dual of a code $C \subset \mathbb{F}_q^N$. For convenience of the reader, we recall that if one writes $\mathbf{x}=(x_{1}, \ldots, x_N)$ a vector in  $\mathbb{F}_{q^2}^{N}$, the Hermitian product of two vectors $\mathbf{x}$ and $\mathbf{y}$ is defined as $\mathbf{x} \cdot_h \mathbf{y} := \sum_{i=1}^{N} x_i y_i^q$, and the Hermitian dual of a code $C \subset \mathbb{F}_{q^2}^{N}$ as $C^{\perp_h} := \{\mathbf{x} \in \mathbb{F}_{q^2}^{N} \; : \; \mathbf{x} \cdot_h \mathbf{y}=0 \mbox{ for all } \mathbf{y} \in C\}$. Let $\#$ stand for cardinality, from \cite{galindo-hernando,QINP2}, one can deduce some important properties of the previously introduced codes:

\begin{pro}
\label{BCH}
The following statements hold.
\begin{enumerate}
\item The code $\mathrm{BCH}_q^N(t,t')$ has length $N$ and dimension $\sum_{\ell=t}^{t'} \# \Lambda^q_{\alpha_\ell}$.
\item The minimum distance $d$ of the Euclidean dual $[\mathrm{BCH}_q^N(1,t)]^{\perp_e}$ (respectively,\\ $[\mathrm{BCH}_q^N(0,t)]^{\perp_e}$) of the code $\mathrm{BCH}_q^N(1,t)$ (respectively, $\mathrm{BCH}_q^N(0,t)$) satisfies $d \geq \alpha_{t+1}$ (respectively, $d \geq \alpha_{t+1} +1$).
\item Suppose that $s=2\varsigma$. Then, the BCH code $$\mathrm{BCH}_{q^2}^N(t,t'): = \mathcal{C}^{N}_{\Delta_{q^2}^{N} (t,t')} \cap \mathbb{F}_{q^2}^N$$ has length $N$ and dimension $\sum_{\ell=t}^{t'} \# \Lambda^{q^2}_{\alpha_\ell}$. Moreover, the minimum distance $d$ of the Hermitian dual $[\mathrm{BCH}_{q^2}^N(1,t)]^{\perp_h}$ (respectively, $[\mathrm{BCH}_{q^2}^N(0,t)]^{\perp_h}$) of the code $\mathrm{BCH}_{q^2}^N(1,t)$ (respectively, $\mathrm{BCH}_{q^2}^N(0,t)$) satisfies $d \geq \alpha_{t+1}$ (respectively, $d \geq \alpha_{t+1} +1$).
\end{enumerate}
\end{pro}

\begin{rem}
\label{BCHq}
{\rm
Notice that  when $s = 2 \varsigma$, one can consider cyclotomic cosets $\Lambda_e^q$ and $\Lambda_e^{q^2}$. Along the paper, whenever we write $s=2\varsigma$, it means that we are going to consider cyclotomic cosets with respect to $q^2$. Moreover, in this case (for instance in Item (3) of Proposition \ref{BCH}), for simplicity, we also write $\mathcal{A}_{q^2}(N) = \{\alpha_0=0 < \alpha_1 < \cdots < \alpha_\nu \}$. Generally speaking $\mathcal{A}_{q}(N) \neq \mathcal{A}_{q^2}(N)$.
}
\end{rem}

\subsection{H-BCH codes}
\label{22}
In \cite{HBCH}, homothetic-BCH codes (H-BCH codes) were introduced. These codes behave in a similar way to subfield-subcodes of $\{1\}$-affine variety codes, but they reach lengths  that a subfield-subcode as before cannot. Let us recall the definition. With the above notation, fix a positive natural number $n$ which divides $N$ and take another positive integer $\lambda < \frac{N}{n}$ such that $\lambda n$ does not divide $N$. Consider the element $\zeta_n = a^{\frac{N}{n}}$ and the set
\begin{equation}
P:=\left\{1,\zeta_{n},\dots,\zeta_{n}^{n-1},a,a\zeta_{n},
\dots,a\zeta_{n}^{n-1},\dots,a^{\lambda-1},a^{\lambda-1}\zeta_{n},\dots,a^{\lambda-1}\zeta_{n}^{n-1}\right\}, \label{eq:pointsetHBCH}
\end{equation}
which is the vanishing set of the ideal of the polynomial ring $\F_{q^{s}}[X]$ $$I=\left\langle(X^{n}-1)(X^{n}-a^{n})\cdots (X^{n}-a^{(\lambda-1)n})\right\rangle.$$
For simplicity, set $P=\{\gamma_1, \ldots, \gamma_{\lambda n}\}$ and, as before, consider the
linear evaluation map
    $$
    \ev^P: \faktor{\F_{q^{s}}[X]}{I} \to \F_{q^{s}}^{n \lambda} \textrm{, } \quad \ev^P(f)=\left(f(\gamma_1),\dots,f(\gamma_{\lambda n})\right).
    $$
It allows us to state the following definition:

\begin{de}
\label{definit}
Let $\Delta \subseteq \mathbb{Z}_{N}$. The $q^s$-ary code of length $\lambda n$
$$\mathcal{H}_\Delta^P:=\left\langle \ev^P(X^e) \; : \; e\in\Delta\right\rangle,$$
where $X^e$ stands also for $X^e + I$, is called the {\it homothetic evaluation code} given by $\Delta$. Its subfield-subcode over $\F_{q}$
\[
\Su_\Delta^{P,q} := \mathcal{H}_\Delta^P \bigcap \F_{q}^{\lambda n}
\]
is named the $q$-ary {\it homothetic-BCH (H-BCH)} code given by $\Delta$.
\end{de}

\begin{rem}
{\rm
When one considers a value $\lambda$ such that $\lambda n$ divides $q^s -1$, one could set $N= \lambda n$ and $\mathcal{H}_\Delta^P$ would be a subfield-subcode of a $\{1\}$-affine variety code. Thus, for getting lengths not provided by $\{1\}$-affine variety codes, we should take values $\lambda$ such that $\lambda n$ does not divide $q^s -1$.
}
\end{rem}

As exposed in \cite[Subsection 2.2]{HBCH}, the following result gives information about H-BCH codes.

\begin{teo}\label{bounds} Keep the above notation and take a nonnegative integer $0 \leq \tau < \nu$.  Denote dimension by $\dim$ and minimum distance by $\dis$, then the following bounds hold:
\begin{enumerate}
    \item $\dim \left(\Su_{\Delta_q^{N} (1,\tau)}^{P,q}\right) \leq \sum_{\ell=1}^\tau \# \Lambda^q_{\alpha_\ell}$;
    $\; \;\dim \left(\Su_{\Delta_q^{N} (0,\tau)}^{P,q}\right) \leq \sum_{\ell=1}^\tau \# \Lambda^q_{\alpha_\ell} + 1$.
    \item $\dis \left( \Su_{\Delta_q^{N} (1,\tau)}^{P,q} \right)^{\perp_e} \geq \alpha_{\tau+1}$;
    $\; \; \dis \left(\Su_{\Delta_q^{N} (0,\tau)}^{P,q}\right)^{\perp_e} \geq \alpha_{\tau+1} + 1$.
    \item Suppose $s=2\varsigma$ and, as explained in Remark \ref{BCHq}, consider sets $\Delta^{N}_{q^2} (t,t')$ obtained as a union of cyclotomic cosets $\Lambda^{q^2}_e$. Then, $\dim \left(\Su_{\Delta_{q^2}^{N} (1,\tau)}^{P,q^2}\right) \leq \sum_{\ell=1}^\tau \# \Lambda^{q^2}_{\alpha_\ell}$,
    $$\; \;\dim \left(\Su_{\Delta_{q^2}^{N} (0,\tau)}^{P,q^2}\right) \leq \sum_{\ell=1}^\tau \# \Lambda^{q^2}_{\alpha_\ell} + 1,$$ and $\dis \left( \Su_{\Delta_{q^2}^{N} (1,\tau)}^{P,q^2} \right)^{\perp_h} \geq \alpha_{\tau+1}$,
    $\; \; \dis \left(\Su_{\Delta_{q^2}^{N} (0,\tau)}^{P,q^2}\right)^{\perp_h} \geq \alpha_{\tau+1} + 1$.
\end{enumerate}
\end{teo}

\begin{rem}
{\rm
Replacing $\Delta_q^{N} (1,\tau)$ (respectively, $\Delta_{q^2}^{N} (1,\tau)$) with a set $\Delta$ which does not contain $0$, it is a union of cyclotomic cosets $\Lambda_e^q$ (respectively, $\Lambda_e^{q^2}$) and contains $\mathfrak{i} -1$ consecutive positive integers, one gets close results to those in Theorem \ref{bounds}. In fact, $\dim \left(\Su_{\Delta}^{P,q}\right) \leq \# \Delta$ and $\dis \left( \Su_{\Delta}^{P,q} \right)^{\perp_e} \geq \mathfrak{i}$ (respectively, $\dim \left(\Su_{\Delta}^{P,q^2}\right) \leq \# \Delta$ and $\dis \left( \Su_{\Delta}^{P,q^2} \right)^{\perp_h} \geq \mathfrak{i}$).
}
\end{rem}

\subsection{Classical \texorpdfstring{$(r,\delta)$}{lr2}-locally recoverable codes}
\label{LRC}
Let $C \subseteq \mathbb{F}_q^m$ be an error-correcting $q$-linear code. Given a subset $T \subseteq \{1, \ldots, m\}$ of cardinality $t$, the {\it puncturing of $C$ at $T$} is the set
$$\pi_T (C) = \{ \mathbf{c}_T :=(c_j)_{j \in T} \; :  \; \mathbf{c}=(c_1, \ldots, c_m) \in C\},$$
where $\pi_T: \mathbb{F}_q^m \rightarrow \mathbb{F}_q^t$ is
the projection map on the coordinates of $T$. The set $\pi_T (C)$ is also named the {\it punctured code} of $C$ at $T$.

In addition,  the set
$$\sigma_T (C) = \{ \mathbf{c}_T\; : \; \mathbf{c}=(c_1, \ldots, c_m) \in C \mbox{ satisfying supp}(\mathbf{c}) \subseteq T \},$$
where supp$(\mathbf{c}) := \{j \in \{1, \ldots, m\} : c_j \neq 0\}$, is called the {\it shortening of $C$ at $T$}. As before, $\sigma_T (C)$ is also named the {\it shortened code} of $C$ at $T$.

By \cite{QINP3}, it holds that
\begin{equation}
  \label{dualpunct}
  \pi_T \left( C^{\perp_e} \right) = \left[ \sigma_T \left( C \right) \right]^{\perp_e}, \mbox{ and also that}
\end{equation}
\begin{equation}
  \label{dualpuncth}
  \pi_T \left( C^{\perp_h} \right) = \left[ \sigma_T \left( C \right) \right]^{\perp_h},
\end{equation}
whenever $C \subseteq \mathbb{F}_{q^2}^m$.

In this paper we are interested in quantum $(r,\delta)$-locally recoverable codes. Thus, it makes sense to start with the definition of classical $(r,\delta)$-locally recoverable code.
\begin{de}
An {\it $(r,\delta)$-locally recoverable code} --$(r,\delta)$-LRC--, also named an LRC with  locality $(r,\delta)$, is
an error-correcting code $C \subseteq \mathbb{F}_q^m$ such that, for each index $i \in \{1, \ldots, m\}$, there exists a set of indices $J \subseteq \{1, \ldots, m\}$, $i \in J$, satisfying:
\begin{enumerate}
\item $\#(J) \leq r + \delta -1$.
\item $\dis[\pi_{J} (C)] \geq \delta$, where $\#$ means cardinality and $\dis$ minimum distance.
\end{enumerate}
\end{de}
The set $J$ is named an $(r,\delta)$-{\it recovery set} for the $i$-th coordinate. From the previous definition, one can deduce that, given any set $I \subsetneq J$ of cardinality $\delta -1$, the coordinates corresponding to $I$ of any word in the code can be recovered from the coordinates corresponding to $J \setminus I$.

In \cite{PKLK2012},  the following Singleton-like bound was introduced for $(r,\delta)$-LRCs.

\begin{pro}
\label{Singleton-2}
Let $[m,k,d]_q$ be the parameters of an $(r,\delta)$-LRC, $C$. Then,
\begin{equation}
\label{Singleform}
k+d+ \left(\left\lceil{\frac{k}{r}} \right\rceil -1 \right)\left( \delta -1 \right) \leq m+1.
\end{equation}
\end{pro}

An $(r,\delta)$-LRC is {\it optimal} whenever it  meets the bound in Proposition \ref{Singleton-2}.



\subsection{Stabilizer quantum codes}
Denote by $\mathbb{C}$ the complex numbers. According to \cite{Ketkar}, a stabilizer (quantum) code is a non-zero subspace of the complex linear space $\mathbb{C}^{q^m}$ of the form
\[
\bigcap_{E \in \Gamma} \left\{ \mathbf{v} \in \mathbb{C}^{q^m} \; : \; E \mathbf{v} = \mathbf{v} \right\},
\]
$\Gamma$ being some abelian subgroup of the error group $G_m$ generated by a nice error basis  on $\mathbb{C}^{q^m}$. The parameters of a stabilizer code $Q$ are usually written as $((m, K, d))_q$, where $K$ is the dimension of $Q$ as a linear space and $d$ its minimum distance, which means that $Q$ allows us to detect all errors in $G_m$ of weight less than $d$, but not some error of weight $d$. In the case when $K=q^k$, the parameters of $Q$ are expressed as $[[m,k,d]]_q$.

Considering the so-called symplectic weight $\mathrm{swt}$ on the linear space $\mathbb{F}_q^{2m}$ and the trace-symplectic dual $C^{\perp_s}$ of an additive code $C$, the existence of an $((m, K, d))_q$ stabilizer code is equivalent to that of an additive code $C \subseteq \mathbb{F}_q^{2m}$ of size $\# C = \frac{q^m}{K}$ such that $C \subseteq C^{\perp_s}$ and $\mathrm{swt} (C^{\perp_s} \setminus C) =d$ if $K>1$ ($\mathrm{swt} (C^{\perp_s}) =d$ in case $K=1$), \cite[Theorem 13]{Ketkar}.

We will only consider linear codes and, therefore, it suffices to use Euclidean and Hermitian inner products, and Euclidean $C^{\perp_e}$ (respectively, Hermitian $C^{\perp_h}$) dual spaces of $q$-linear $C \subseteq \mathbb{F}_q^{m}$ (respectively, $q^2$-linear $C \subseteq \mathbb{F}_{q^2}^{m}$) spaces.

In this context, denoting by $\mathrm{wt}$ the Hamming weight, we conclude this subsection by stating the main result for obtaining stabilizer codes from linear codes.
\begin{teo}
\label{resto}
The following statements hold:\\
1. Let $C$ be a Hermitian dual-containing  linear $[m,k,d]_{q^2}$-code (that is, $C$ satisfies $C^{\perp_h} \subseteq C$), then there exists an $[[m, 2k-m, \geq d]]_q$ stabilizer quantum  code.\\
2. Let $C_j$  be two $q$-ary linear codes, $j=1, 2$, with parameters $[m,k_j,d_j]_q$ and assume that $C_2^{\perp_e} \subseteq C_1$.  Then, there exists an $[[m, k_1+k_2-m,d]]_q$ stabilizer quantum  code with $d= \min \left\{\mathrm{wt}(\mathbf{z}) : \mathbf{z} \in (C_1 \setminus C_2^{\perp_e}) \cup (C_2 \setminus C_1^{\perp_e})\right\}$.
\end{teo}

\subsection{Quantum \texorpdfstring{$(r,\delta)$}{lr4}-locally recoverable codes}
\label{sect25}
Generalizing concepts introduced in \cite{Golowich, Golowich2}, quantum $(r,\delta)$-locally recoverable codes --quantum $(r,\delta)$-LRCs-- were introduced in \cite{QLRC}. Before giving the definition, we need to introduce the meaning of local recoverability for pairs of sets of indices.

Assume we have a quantum code $Q \subseteq \mathbb{C}^{q^m}$ and two subsets $\emptyset \neq I \subsetneq J \subseteq \{1, \ldots, m\}$. Set $\Gamma$ the complete depolarizing channel on a simple qudit in $\mathbb{C}^q$ and define $\Gamma^I = \bigotimes_{i=1}^m \Gamma_i$, where $\Gamma_i = \Gamma$ when $i \in I$ and $\Gamma_i$ is the $q \times q$ identity matrix, otherwise. We say that $Q$ is an {\it $(I,J)$-locally recoverable code} whenever there is  a trace-preserving quantum operation $\mathcal{R}_{Q,I}^{J}$, acting only on the qudits corresponding to $J$, that satisfies
\begin{equation*}
  \mathcal{R}_{Q,I}^J \circ \Gamma^I (\ket{\varphi}\bra{\varphi}) = \ket{\varphi}\bra{\varphi} \label{eq2}
\end{equation*}
for all $\ket{\varphi} \in Q$.

Then, the aforementioned definition is the following.

\begin{de}
\label{lcr}
A quantum code  $Q \subseteq \mathbb{C}^{q^m}$ is a {\it quantum $(r,\delta)$-locally recoverable code}  if for each index $i \in \{1, \ldots, m\}$, there exists a set $J \subseteq \{1, \ldots, m\}$ containing $i$ with $\# (J) \leq r + \delta-1$ satisfying that for any subset $I \subseteq J$ of cardinality $\delta -1$, $Q$ is $(I,J)$-locally recoverable.
\end{de}

When the quantum code is a stabilizer one and comes from a Hermitian or Euclidean self-orthogonal linear code $C$ (see, Theorem \ref{resto} where one should change $C$ by its Hermitian or Euclidean dual), we denote it by $Q'(C)$ and the following characterizing result holds (see \cite{QLRC}).

\begin{teo}
\label{quantumH-E}
The following characterizations hold:
\begin{enumerate}
\item  Let $Q'(C) \subseteq \mathbb{C}^{q^m}$ be a stabilizer code given by a $q^2$-ary Hermitian self-orthogonal code $C \subseteq \mathbb{F}_{q^2}^m$. Then $Q'(C)$ is a quantum $(r,\delta)$-LRC if and only if, for every index $i \in \{1, \ldots, m\}$, there exists a set $J \ni i$ of cardinality $\leq r + \delta -1$ satisfying the following equality for every subset  $I \subset J$ with $|I| = \delta-1$:
\begin{equation*}
\sigma_I\left[\pi_J \left(C^{\perp_h}\right)\right] = \sigma_I(C).
\end{equation*}
\item Let $Q'(C) \subseteq \mathbb{C}^{q^m}$ be a stabilizer code given by a $q$-ary Euclidean self-orthogonal code $C \subseteq \mathbb{F}_{q}^m$. Then $Q'(C)$ is a quantum $(r,\delta)$-LRC if and only if, for every index $i \in \{1, \ldots, m\}$, there exists a set $J \ni i$ of cardinality $\leq r + \delta -1$ satisfying the following equality for every subset $I \subset J$ with $|I| = \delta-1$:
\begin{equation*}
\sigma_I\left[\pi_J \left(C^{\perp_e}\right)\right] = \sigma_I(C).
\end{equation*}
\end{enumerate}
\end{teo}

Concerning Singleton-like bounds for our  quantum $(r,\delta)$-LRCs, again in \cite{QLRC}, it is proved the following theorem which considers quantum codes $Q(C)$ coming from linear codes $C$ as described in Theorem \ref{resto}.

\begin{teo}
\label{SingletonQ}
Let $C$ be a $q^2$-ary  linear code included in $\mathbb{F}_{q^2}^m$ (respectively, a  $q$-ary linear code included in $\mathbb{F}_{q}^m$). Assume that $C$ is Hermitian  (respectively, Euclidean) dual-containing, $\dim C = \frac{m+k}{2}$ and  $C$ is an $(r,\delta)$-LRC. Then, the quantum $(r,\delta)$-LRC coming from $C$, $Q(C)$, has parameters $[[m,k, \geq \dis(C)]]_q$, which satisfy
\begin{equation}
\label{eq17}
k + 2 \dis(C) + 2\left(\left\lceil\frac{m+k}{2r}\right\rceil-1\right)(\delta-1) \leq m+2.
\end{equation}
\end{teo}

Notice that Inequality (\ref{eq17}) depends on the minimum distance $\dis(C)$ of the code $C$ and not on the minimum distance of $Q(C)$, but for pure codes $Q(C)$, it can be written:
\begin{equation}
\label{PURE}
k + 2 \dis(Q(C)) + 2\left(\left\lceil\frac{m+k}{2r}\right\rceil-1\right)(\delta-1) \leq m+2.
\end{equation}
As a consequence, a concept of optimal pure quantum $(r,\delta)$-LRCs can be given for  stabilizer quantum codes coming from linear codes:
\begin{de}
{\rm A pure quantum $(r, \delta)$-LRC,  $Q(C)$, is said to be {\it optimal} whenever its parameters and $(r, \delta)$-locality attain the bound in (\ref{PURE}).}
\end{de}

\section{Quantum locally recoverable codes given by BCH and homothetic-BCH codes}
\label{LA3}

BCH codes are a very useful class of classical error-correcting codes. In this section, we explain how quantum $(r,\delta)$-LRCs can be constructed from BCH and H-BCH codes. Theorem \ref{SingletonQ} shows that we only need to construct quantum codes $Q(C)$ given by $(r,\delta)$-LRCs $C$.


\subsection{Local recoverability}
\label{Subs31}
In this subsection we explain how dual codes of subfield-subcodes of $\{1\}$-affine variety codes and of H-BCH codes can be regarded as $(r, \delta)$-LRCs. Keep the above notation and consider the set
$$P_{\mathfrak{A}}:=\left\{1,\zeta_{n},\dots,\zeta_{n}^{n-1},a,a\zeta_{n},
\dots,a\zeta_{n}^{n-1},\dots,a^{\frac{N}{n}-1},a^{\frac{N}{n}-1}\zeta_{n},\dots,a^{\frac{N}{n}-1}\zeta_{n}^{n-1}\right\},$$
that is the above defined set $P$ but for $\lambda = \frac{N}{n}$. Then, one can identify the code $\mathcal{H}_\Delta^{P_{\mathfrak{A}}}$, as introduced in Definition \ref{definit}, with $\mathcal{C}_\Delta^N$ because we have only permuted the symbols in each word by  the same permutation. Note that $\mathcal{H}_\Delta^{P_{\mathfrak{A}}}$ is a $\{1\}$-affine variety code, but not an H-BCH code because $\lambda n$ divides $N$. From now on, we will consider evaluation maps on $P_{\mathfrak{A}}$ or on some of its subsets.



Given a set $D \subseteq \mathbb{Z}_N$, recall that its elements are representatives of the elements in the congruence ring $\mathbb{Z}/N\mathbb{Z}$. For any positive divisor $M$ of $N$, we define $D_M$ as the subset of $\mathbb{Z}_M$ given by the representatives of the classes in $\mathbb{Z}/M\mathbb{Z}$ of the elements in $D$; clearly if $D \subseteq \mathbb{Z}_M$, $D_M = D$. Also, define
\[
- D_M := \left\{ M-i \, \mod M \; : \; i \in D \right\}.
\]
And when $s= 2 \varsigma$, set $- D_M^q := \left\{ M- qi \, \mod M \; : \; i \in D \right\}.$

As above, write $\mathcal{S}_\Delta^{P_{\mathfrak{A}},q}:= \mathcal{H}_\Delta^{P_{\mathfrak{A}}} \cap \mathbb{F}_q^N$. The next result allows us to regard $\left(\mathcal{S}_\Delta^{P_{\mathfrak{A}},q}\right)^{\perp_e}$ as an $(r,\delta)$-LRC whenever suitable sets $\Delta$ are chosen. In the subsequent result we will see a similar result for $\left(\mathcal{S}_\Delta^{P_{\mathfrak{A}},q^2}\right)^{\perp_h}$ when $s=2 \varsigma$. Both are easy consequences of a very slight generalization of Theorem 12 in \cite{QZF2021}. First we need to give a definition.

\begin{de}
\label{def13}
{\rm
Let $M$ be a positive integer which divides $N$. A subset $D$ of $\mathbb{Z}_M$ is said to be {\it $q$-complete modulo $M$} if it is a union of cyclotomic cosets $\Lambda^q_e$ in $\mathbb{Z}_M$. When $s= 2 \varsigma$, $D$ is {\it $q^2$-complete modulo $M$} whenever it is a union of cyclotomic cosets $\Lambda^{q^2}_e$.
}
\end{de}

\begin{pro}
\label{era1}
Let $A$ and $B$ be subsets of $\mathbb{Z}_N$. Suppose that $$A+B := \{a+b \; : \; a \in A, \; b\in B\} \subseteq \Delta,$$ $\Delta$ being a $q$-complete modulo $N$ subset of $\mathbb{Z}_N$. Denote by $d_B^{\perp_e}$  the minimum distance of the code $\left(\mathcal{H}_B^{P_{\mathfrak{A}}}\right)^{\perp_e}$  and by $d_A$ the minimum distance of $\mathcal{H}_A^{P_{\mathfrak{A}}}$. Assume that $\delta$ is a positive integer such that $\delta \leq \min\{d_A,  d_B^{\perp_e}\}$.
Then, the code $\left(\mathcal{S}_\Delta^{P_{\mathfrak{A}},q}\right)^{\perp_e}$  is a $q$-ary $(r,\delta)$-LRC with locality $(d_A - \delta +1, \delta)$.
\end{pro}
\begin{proof}
Consider the code $\mathcal{H}_\Delta^{P_{\mathfrak{A}}} \subseteq \mathbb{F}_{q^s}^N$. Its Euclidean dual $\left(\mathcal{H}_\Delta^{P_{\mathfrak{A}}}\right)^{\perp_e}$ is a cyclic code $\mathrm{Cy}(\Delta)$ generated by $\prod_{\ell \in \Delta} (X-a^\ell)$, $a$ being as in Subsection \ref{BCH-1}. Now $\mathcal{H}_A^{P_{\mathfrak{A}}}= \left( \mathrm{Cy}(A)\right)^{\perp_e}$ and  $\left(\mathcal{H}_B^{P_{\mathfrak{A}}}\right)^{\perp_e} = \mathrm{Cy}(B)$. Following the proof of \cite[Theorem 12]{QZF2021} and, with the notation in that proof, considering $\delta -1$ columns (instead of $d_B -1$) of the matrix corresponding to the cyclic code $\overline{\mathcal{C}}_B$ used in \cite[Theorem 12]{QZF2021}, it is straightforward to deduce that the subfield-subcode $\left(\mathcal{H}_\Delta^{P_{\mathfrak{A}}}\right)^{\perp_e} \cap \mathbb{F}_{q}^N$ is an $(r,\delta)$-LRC with locality $(d_A - \delta +1, \delta)$. In our conditions, $\left(\mathcal{H}_\Delta^{P_{\mathfrak{A}}}\right)^{\perp_e} \cap \mathbb{F}_{q}^N = \left(\mathcal{S}_\Delta^{P_{\mathfrak{A}},q}\right)^{\perp_e}$, which concludes the proof.

To see the above equality of codes, it suffices to consider the following chain of equalities:
\begin{equation}
\label{EE}
\left(\mathcal{H}_\Delta^{P_{\mathfrak{A}}}\right)^{\perp_e} \cap \mathbb{F}_{q}^N =
\left(\mathcal{S}_{\Delta^{\perp_e}}^{P_{\mathfrak{A}},q}\right) =
\mathbf{tr} \left(\mathcal{H}_{\Delta^{\perp_e}}^{P_{\mathfrak{A}}}\right) =
\left(\mathcal{S}_\Delta^{P_{\mathfrak{A}},q}\right)^{\perp_e},
\end{equation}
where $\mathbf{tr}$ is given by the trace map from $\mathbb{F}_{q^s}$ to $\mathbb{F}_{q}$ applied componentwise and $\Delta^{\perp_e} = \mathbb{Z}_N \setminus -\Delta_N$. The first equality follows from \cite[Proposition 3]{galindo-hernando}, which proves that $\left(\mathcal{H}_\Delta^{P_{\mathfrak{A}}}\right)^{\perp_e} = \left(\mathcal{H}_{\Delta^{\perp_e}}^{P_{\mathfrak{A}}}\right)$, and the second and third ones are true by \cite[Theorem 4]{galindo-hernando} and the proof of Theorem 5 in \cite{galindo-hernando}.

\end{proof}

Our next result assumes $s= 2\varsigma$ and writes $\mathfrak{p} := q^{\varsigma}$ and for a set $\Delta \subseteq \mathbb{Z}_N$, we regard $\mathcal{H}_\Delta^{P_{\mathfrak{A}}}$ as a code over $\mathbb{F}_{q^{s}}$ but also over $\mathbb{F}_{\mathfrak{p}^{2}}$. We write $\left(\mathcal{H}_\Delta^{P_{\mathfrak{A}}}\right)^{\perp_h}$ the Hermitian dual of this last code.

\begin{pro}
\label{era1-B}
Keep the above notation and suppose $s= 2\varsigma$. Let $A$ and $B$ be subsets of $\mathbb{Z}_N$. Suppose that $$A+B := \{a+b \; : \; a \in A, \; b\in B\} \subseteq \Delta,$$ $\Delta$ being a $q^2$-complete modulo $N$ subset of $\mathbb{Z}_N$. Denote by $d_B^{\perp_{h}}$  the minimum distance of the code $\left(\mathcal{H}_B^{P_{\mathfrak{A}}}\right)^{\perp_h}$  and by $d_A$ the minimum distance of $\mathcal{H}_A^{P_{\mathfrak{A}}}$. Assume that $\delta$ is a positive integer such that $\delta \leq \min\{d_A,  d_B^{\perp_{h}}\}$.
Then, the code $\left(\mathcal{S}_\Delta^{P_{\mathfrak{A}},{q^2}}\right)^{\perp_{h}}$  is a $q^2$-ary $(r,\delta)$-LRC with locality $(d_A - \delta +1, \delta)$.
\end{pro}

\begin{proof}
We reason in a similar way as we did in the proof of Proposition \ref{era1}. The Hermitian dual $\left(\mathcal{H}_\Delta^{P_{\mathfrak{A}}}\right)^{\perp_h}$ is a cyclic code $\mathrm{Cy}^q(\Delta)$. The set $\Delta$ is $q^2$-complete modulo $N$ and, applying again \cite[Theorem 12]{QZF2021}, one gets that the subfield-subcode $\left(\mathcal{H}_\Delta^{P_{\mathfrak{A}}}\right)^{\perp_h} \cap \mathbb{F}_{q^2}^N$ is an $(r,\delta)$-LRC with locality $(d_A - \delta +1, \delta)$. Note that this holds because Hermitian and Euclidean dual codes are isometric.

As in the proof of Proposition \ref{era1}, it implies that $\left(\mathcal{S}_\Delta^{P_{\mathfrak{A}},q^2}\right)^{\perp_h}$ is an $(r,\delta)$-LRC with the previous given locality. To see the corresponding equality of codes, it suffices to notice that the chain of equalities (\ref{EE}) is true for $\perp_h$ instead of $\perp_e$ and $q^2$ instead of $q$. To prove it, we should note that $\left(\mathcal{H}_\Delta^{P_{\mathfrak{A}}}\right)^{\perp_h} = \left(\mathcal{H}_{\Delta^{\perp_h}}^{P_{\mathfrak{A}}}\right)$, where $\Delta^{\perp_h} = \mathbb{Z}_N \setminus -\Delta_N^q$ by \cite[Proposition 4]{QINP} and use the proof of \cite[Theorem 7]{QINP} and Theorem 4 in \cite{galindo-hernando}.
\end{proof}

To apply Propositions \ref{era1} and \ref{era1-B}, it is convenient to determine the minimum distance $d_A$ or, at least, a bound of it. The following proposition gives this minimum distance for specific sets $A$.

\begin{pro}
\label{era3}
Keep the previous conditions where $N$ and $n$ are positive integers such that $n$ divides $N$ which divides $q^s -1$. Consider the set
\[
\mathcal{A} := \{0, n, 2n, \ldots, N-n\} \subset \mathbb{Z}_N.
\]
Then, $\dis \left(\mathcal{H}_{\mathcal{A}}^{P_{\mathfrak{A}}}\right) = n$.
\end{pro}
\begin{proof}
Take the polynomial in $\mathbb{F}_{q^s} [X]$
\[
f= 1+ X^n + X^{2n} + \cdots + X^{N-n}.
\]
The facts that the degree of $f$ is $N-n$ and $ f= \frac{X^N -1}{X^n -1}$ prove that all the elements in $P_{\mathfrak{A}}$, with the exception of $1, \zeta_n, \ldots, \zeta_n^{n-1}$, are roots of $f$, therefore the Hamming weight $\mathrm{wt}\left(\mathrm{ev}^{P_{\mathfrak{A}}} (f)\right) =n$. As a consequence, the minimum distance $\dis \left( \mathcal{H}_{\mathcal{A}}^{P_{\mathfrak{A}}}\right) \leq n$.

To conclude the proof, it remains to show that $\dis \left( \mathcal{H}_{\mathcal{A}}^{P_{\mathfrak{A}}}\right) \geq n$. By Proposition 3 in \cite{galindo-hernando}, $ [\mathcal{H}_{\mathcal{A}}^{P_{\mathfrak{A}}}]^{\perp_e} = \mathcal{H}_{\{1, \ldots, n-1, n+1, \ldots, 2n-1, \ldots, N-n-1\}}^{P_{\mathfrak{A}}}$. Then,
\[
\dis \left(\mathcal{H}_{\mathcal{A}}^{P_{\mathfrak{A}}}\right) = \dis \left[\left(\mathcal{H}_{\mathcal{A}}^{P_{\mathfrak{A}}}\right)^{\perp_e}\right]^{\perp_e} \geq n,
\]
which concludes the proof.
\end{proof}

\begin{rem}
{\rm
The concept of $q$-complete modulo $M$ set introduced in Definition \ref{def13} can be extended to any positive integer $z$ (not necessarily prime) relatively prime to $M$. Thus, a subset $S \subseteq \mathbb{Z}_M$ is {\it $z$-complete modulo $M$} whenever $\sigma z \in S$ for all $\sigma \in S$.
}
\end{rem}

The above set $\mathcal{A}$ will be instrumental in several of our forthcoming results directed at find localities of the dual codes of the previously introduced subfield-subcodes, regarded as $(r, \delta)$-LRCs.

Let us prove some simple properties of $\mathcal{A}$.

\begin{pro}
\label{1-2-jap}
Keep the above notation where $N$ and $n$ are positive integers such that $n$ divides $N$ and $N$ divides $q^s -1$. Let $z$ be a positive integer relatively prime to $N$.
\begin{enumerate}
\item The set $\mathcal{A}$ introduced above is $z$-complete modulo $N$.
\item Let $B$ be a $z$-complete modulo $n$ set, then $\mathcal{A} + B$ is $z$-complete modulo $N$.
\end{enumerate}
\end{pro}

\begin{proof}

(1) Since $n$ divides $N$, the product of $z$ by a multiple of $n$ modulo $N$ belongs to $\mathcal{A}$.

(2) Consider an element $in +b \in \mathcal{A} + B$, $0 \leq i \leq \frac{N}{n}-1$. Set $b' \in B$ such that $zb \equiv b' \mod n$. Then, $z (in+b) \equiv i' n + b' \mod N$ for a suitable choice of a nonnegative integer $i'$. Therefore, $z (in+b) \in \mathcal{A} + B$ and, then, this set is $z$-complete modulo $N$.
\end{proof}

With the previous notations, denote $\mathcal{S}_\Delta^{P_{\mathfrak{A}}, q}$ the subfield-subcode over the finite field $\mathbb{F}_q$ of the $\{1\}$-affine variety code $\mathcal{H}_\Delta^{P_{\mathfrak{A}}}$. Analogously, if $s = 2\varsigma$, we set $\mathcal{S}_\Delta^{P_{\mathfrak{A}}, q^2}$ the subfield-subcode of the same code but over $\mathbb{F}_{q^2}$.
\begin{teo}
\label{eval}
The following statements hold:
\begin{enumerate}
\item Let $d \leq n $ be a positive integer and $B \subseteq \mathbb{Z}_N$  such that $B$ contains $d-1$ consecutive positive integers. Assume that the set $\Delta:= \mathcal{A} + B \subseteq \mathbb{Z}_N$ is $q$-complete modulo $N$ and let $\delta$ be an integer such that $2 \leq \delta \leq d$. Then, the code $\left(
\mathcal{S}_\Delta^{P_{\mathfrak{A}}, q}
\right)^{\perp_e}$ is an $( n - \delta +1, \delta)$-locally recoverable code  with parameters $[N, N - \# \Delta, \geq d]_q$.

\item Suppose that $s = 2 \varsigma$. Let $B \subseteq \mathbb{Z}_N$  such that it contains $d-1$ consecutive positive integers. Assume that $\Delta:= \mathcal{A} + B \subseteq \mathbb{Z}_N$  is a $q^2$-complete modulo $N$ set. Let $\delta$ and $d$ as Item (1). Then, the code $\left(\mathcal{S}_\Delta^{P_{\mathfrak{A}}, q^2}\right)^{\perp_h}$ is an $( n - \delta +1,  \delta)$-locally recoverable code  with parameters $[N, N - \# \Delta, \geq d]_{q^2}$.
\end{enumerate}
\end{teo}
\begin{proof}
Let us prove Item (1). When one considers a $q$-complete modulo $N$ set $\Omega$, by \cite[Theorem 5]{QINP}, the dimension of the subfield-subcode $\mathcal{S}_\Omega^{P_{\mathfrak{A}}, q}$ equals $\# \Omega$. The bounds for the minimum distances hold because the parity-check matrix of the code $\left(\mathcal{H}_\Delta^{P_{\mathfrak{A}}}\right)^{\perp_e}$ contains a Vandermonde matrix giving a bound for its distance which  $\left(
\mathcal{S}_\Delta^{P_{\mathfrak{A}}, q}
\right)^{\perp_e}$ inherits. Details can be seen in the proof of \cite[Proposition 3.6]{QINP2}. Finally, the facts that $ \dis\left(\mathcal{H}_{\mathcal{A}}^{P_{\mathfrak{A}}}\right) =n$ and $\dis \left(\mathcal{H}_{B}^{P_{\mathfrak{A}}}\right)^{\perp_e} \geq \delta$ imply that $\delta \leq \min\{d_A, d_B^{\perp_e}\}$ and by Proposition \ref{era1} we get the locality of the code.

A proof for Item (2) follows analogously after noting that we are considering $q^2$-complete  modulo $N$ sets $\Omega$, see for instance the proof of \cite[Proposition 3.10]{QINP2}. Here, we have to use Proposition \ref{era1-B} for giving the locality of $\left(\mathcal{S}_\Delta^{P_{\mathfrak{A}}, q^2}\right)^{\perp_h}$.
\end{proof}

Proposition \ref{1-2-jap} supports the following corollary of Theorem \ref{eval}. Note that, in this corollary, the elements in $B$ are representatives in $\mathbb{Z}_n$ and we consider $n-1$ consecutive to $0$.

\begin{cor}
\label{mats-1}
Theorem \ref{eval} (1) (respectively, (2)) holds for $\Delta = \mathcal{A} + B$, $B$ formed by representatives in $\mathbb{Z}_n$ and being $q$ (respectively, $q^2$) complete modulo $n$.
\end{cor}
\begin{proof}
Let us prove the first assertion, the second one follows analogously.

Let $\tilde{B}$ be the smallest $q$-complete modulo $N$ set such that $B \subseteq \tilde{B}$. Then, $\mathcal{A} + B  = \mathcal{A} + \tilde{B}$. Indeed,  clearly $\mathcal{A} + B  \subseteq \mathcal{A} + \tilde{B}$ and we need to prove the opposite inclusion. Let $a + \tilde{b} \in \mathcal{A} + \tilde{B}$, where $a \in \mathcal{A}$ and $\tilde{b} \in \tilde{B}$. Now $\tilde{b} \equiv q^i b \mod N$ for some nonnegative integer $i$ and $b \in B$. Since $\mathcal{A}$ is $q$-complete modulo $N$, there exists $a' \in \mathcal{A}$  such that $q^i a' \equiv a \mod N$. Therefore $q^i (a'+b) \equiv a + \tilde{b} \mod N$ and the fact that $\mathcal{A}$ is $q$-complete modulo $N$ shows that $a + \tilde{b} \in \mathcal{A} + B$ and the inclusion holds. This finishes the proof because $\tilde{B}$ contains $d -1$ consecutive positive integers whenever $B$ fulfils the same property.
\end{proof}

Another corollary of Theorem \ref{eval}, written in terms of BCH codes, is the following one.

\begin{cor}
\label{coreval}
Let $\mathcal{A}$ and $B$ be two sets as  in the statement of Theorem \ref{eval}. With the notation as in Subsection \ref{BCH-1}, suppose that $\mathcal{A}+B$ equals either $\Delta_q^N(0,\tau)$, or $\Delta_{q^2}^N(0,\tau)$ whenever $s= 2\varsigma$, for some $1 \leq \tau < \nu$. Assume also that $2 \leq \delta \leq \alpha_{\tau +1} +1 \leq n$. Then,
\begin{enumerate}
\item The code $\left(\mathrm{BCH}_q^N(0,\tau)\right)^{\perp_e}$ is an $( n - \delta +1, \delta)$-locally recoverable code with parameters $[N, N - \sum_{\ell=0}^{\tau} \# \Lambda_{\alpha_\ell}^q, \geq \alpha_{\tau +1}+1]_q$.
\item If  $s = 2\varsigma$, the code $\left(\mathrm{BCH}_{q^2}^N(0,\tau)\right)^{\perp_h}$ is an $(n - \delta +1, \delta)$-locally recoverable code with parameters $[N, N - \sum_{\ell=0}^{\tau} \# \Lambda_{\alpha_\ell}^{q^2}, \geq \alpha_{\tau +1}+1]_{q^2}$.
\end{enumerate}
\end{cor}

An analog of Theorem \ref{eval} also holds for H-BCH codes. Indeed, suppose that we consider a value $\lambda$ such that $\lambda n$ does not divide $N$ and assume that $B$ is a set as in Corollary \ref{mats-1}. Thus, $\Delta := \mathcal{A} + B$ is $q$-complete (respectively, $q^2$-complete) modulo $N$ as in Theorem \ref{eval}. Then, the result follows from the fact that the code $\mathcal{S}_{\Delta}^{P, q}$ (respectively, $\mathcal{S}_{\Delta}^{P, q^2}$ if $s=2\varsigma$) is a punctured code of $\mathcal{S}_\Delta^{P_{\mathfrak{A}}, q}$ (respectively, $\mathcal{S}_\Delta^{P_{\mathfrak{A}}, q^2}$ if $s=2\varsigma$) and by the relation between punctured and shortened codes given in (\ref{dualpunct}) (respectively, in (\ref{dualpuncth}) if $s=2\varsigma$), the code $\left[ \mathcal{S}_{\Delta}^{P, q} \right]^{\perp_e}$ (respectively,  $\left[ \mathcal{S}_{\Delta}^{P, q^2} \right]^{\perp_h}$ if $s=2\varsigma$) is a shortened code of $\left[ \mathcal{S}_\Delta^{P_{\mathfrak{A}, q}} \right]^{\perp_e}$ (respectively,  $\left[ \mathcal{S}_\Delta^{P_{\mathfrak{A}, q^2}} \right]^{\perp_h}$ if $s=2\varsigma$). Therefore,  locality and minimum distance of $\left[ \mathcal{S}_\Delta^{P, q} \right]^{\perp_e}$ (respectively, $\left[ \mathcal{S}_\Delta^{P, q^2} \right]^{\perp_h}$ if $s=2\varsigma$) are at least as good as those of $\left[ \mathcal{S}_\Delta^{P_{\mathfrak{A}, q}} \right]^{\perp_e}$ (respectively, $\left[ \mathcal{S}_\Delta^{P_{\mathfrak{A}, q^2}} \right]^{\perp_h}$ if $s=2\varsigma$). Note that the locality is deduced from the proof of Theorem 12 in \cite{QZF2021}. In our context, we only need to consider nonzero columns of matrices obtained from the shortened code, where the original codewords contain zeros outside the designated index positions. Let us state the result:




\begin{teo}
\label{mats-2}

Keep the above notation, where $\lambda$ is a positive integer such that $\lambda n$ does not divide $N$.

\begin{enumerate}
\item
Suppose that $\Delta = \mathcal{A} + B$, $B$ being $q$ complete modulo $n$. Assume also that $B$ contains $d-1$ consecutive positive integers. Let $\delta$ be a positive integer such that  $2 \leq \delta \leq d \leq n$. Then, the code $\left(
\mathcal{S}_\Delta^{P, q}
\right)^{\perp_e}$ is an $(n - \delta +1, \delta)$-locally recoverable code  with parameters $[\lambda n,  \lambda (n - \#B), \geq d]_q$.

\item
Suppose that $s = 2 \varsigma$. Suppose that $\Delta = \mathcal{A} + B$, $B$ being $q^2$ complete modulo $n$. Assume also that $B$ contains $d-1$ consecutive positive integers. Let $\delta$ and $d$ be positive integers such that  $2 \leq \delta \leq d$. Then, the code $\left( \mathcal{S}_\Delta^{P, q^2} \right)^{\perp_h}$ is an $(n - \delta +1, \delta)$-locally recoverable code  with parameters $[\lambda n, \lambda (n - \#B), \geq d]_{q^2}$.
\end{enumerate}

\end{teo}

\begin{proof}
Let us prove (1). Item (2) follows analogously.

Considering the reasoning prior to the statement of the theorem, it only remains to calculate the dimension. Consider the sets $I_j := \{jn + 1, \ldots, jn+n\}$, $0 \leq j \leq \frac{N}{n}-1$, which are a partition of $\{1,\dots,N\}$. Set $L:= \cup_{j=0}^{\lambda -1} I_j$ and $K := \cup_{j=\lambda}^{\frac{N}{n} -1} I_j$. It is clear that $\mathcal{S}_{\Delta}^{P,q} = \pi_L\left(\mathcal{S}_{\Delta}^{P_{\mathfrak{A}}, q}\right)$. By \cite[Proposition 1.5.8]{HMCthesis}, it holds that
\[
\mathcal{S}_{\Delta}^{\{a^j, a^j\zeta_n, a^j\zeta_n^2, \ldots,  a^j\zeta_n^{n-1}\},q} = \mathcal{S}_{B}^{\{a^j, a^j\zeta_n, a^j\zeta_n^2, \ldots,  a^j\zeta_n^{n-1}\},q}
\]
because $\mathcal{A}+B$ is congruent to $B$ modulo $n$. As a consequence, $\dim \pi_{I_j}\left(\mathcal{S}_{\Delta}^{P_{\mathfrak{A}}, q}\right) = \#B$.

Now, the fact that
$$\mathcal{S}_{\Delta}^{P_{\mathfrak{A}}, q} \subseteq \pi_L \left( \mathcal{S}_{\Delta}^{P_{\mathfrak{A}}, q}  \right) \times \pi_{K} \left( \mathcal{S}_{\Delta}^{P_{\mathfrak{A}}, q} \right)$$
proves that
\begin{equation}
\label{UNO}
\frac{N}{n} \#B = \dim \left(\mathcal{S}_{\Delta}^{P_{\mathfrak{A}}, q} \right) \leq \dim \left( \pi_L \left( \mathcal{S}_{\Delta}^{P_{\mathfrak{A}}, q} \right) \right) + \dim \left( \pi_K \left( \mathcal{S}_{\Delta}^{P_{\mathfrak{A}}, q} \right) \right).
\end{equation}
The same reasoning allows us to state that
\begin{equation}
\label{DOS}
\dim \left( \pi_L \left( \mathcal{S}_{\Delta}^{P_{\mathfrak{A}}, q} \right) \right) \leq \sum_{j=0}^{\lambda -1} \dim \left(\pi_{I_j} \left( \mathcal{S}_{\Delta}^{P_{\mathfrak{A}}, q} \right) \right) = \lambda \#B
\end{equation}
	and
\begin{equation}
\label{TRES}
\dim \left( \pi_K \left( \mathcal{S}_{\Delta}^{P_{\mathfrak{A}}, q} \right) \right) \leq \sum_{j=\lambda}^{\frac{N}{n} -1} \dim \left(\pi_{I_j} \left( \mathcal{S}_{\Delta}^{P_{\mathfrak{A}}, q} \right) \right) = \left(\frac{N}{n}- \lambda\right) \#B.
\end{equation}	
Combining (\ref{UNO}), (\ref{DOS}) and (\ref{TRES}), one gets the desired equality.
\end{proof}

\subsection{Quantum locally recoverable codes}
\label{sect32}
In this subsection we combine previous results in order to construct quantum $(r,\delta)$-locally recoverable codes. We keep the notation as in Subsection \ref{Subs31}. Recall that $n$ divides $N$ which divides $q^s -1$. The following result will be useful. We continue to consider the before fixed set $\mathcal{A}$.

\begin{pro}
\label{era5}
The following statements hold:

1) If $B \subseteq \mathbb{Z}_n$ is $q$-complete modulo $n$ and $B \cap \left(-B\right) = \emptyset$, then the following intersection of sets in $\mathbb{Z}_N$
\[
 \left(\mathcal{A}+ B \right)_N \cap \left[ - \left(\mathcal{A}+B\right)_N \right]
\]
is empty.

2) If $s= 2 \varsigma$, $B \subseteq \mathbb{Z}_n$ is $q^2$-complete modulo $n$ and $B \cap \left(-B^q\right) = \emptyset$, then the following intersection of sets in $\mathbb{Z}_N$
\[
 \left( \mathcal{A}+ B \right)_N \cap \left[- \left(\mathcal{A}+B\right)^q_N \right]
\]
is empty.
\end{pro}
\begin{proof}
Let us prove 1). The proof of 2) is very similar. Let us assume $j \in \left( \mathcal{A}+ B \right)_N \cap \left[ - \left(\mathcal{A}+B\right)_N \right]$ and we will deduce a contradiction. Taking into account that any element in $\mathcal{A}$ is a multiple of $n$, if $j \equiv a+b \mod N$, $j \in \mathbb{Z}_N$, $a \in \mathcal{A}, b \in B$, then $j \mod n \equiv b \mod n  \in B_n = B$. In addition $j= N- \left( a' + b' \right) \mod N$, $a' \in \mathcal{A}$ and $b' \in B$. Clearly, $N = \frac{N}{n} n$ and $a' = i n$, $i < \frac{N}{n}$. Then, $j=i'n-b'$, and thus $j \mod n \in -B_n =-B$, which contradicts the hypothesis.
\end{proof}

Recall that $P_n := \{1,\zeta_{n},\dots,\zeta_{n}^{n-1} \}$. The key result here is that, for $\mathcal{A}$ and $B$ as above,  Proposition \ref{era5} shows that the Euclidean (respectively, Hermitian when $s=2 \varsigma$) self-orthogonality of the code $\mathcal{S}_{B}^{P_{n}, q}$ (respectively, $\mathcal{S}_{B}^{P_{n}, q^2}$) implies Euclidean (respectively, Hermitian) self-orthogonality of the code $\mathcal{S}_{\mathcal{A}+B}^{P_{\mathfrak{A}}, q}$ (respectively, $\mathcal{S}_{\mathcal{A}+B}^{P_{\mathfrak{A}}, q^2}$). This follows from
\cite[Theorem 5(2)]{QINP} (respectively,
\cite[Theorem 7(1)]{QINP}), because
$ \mathcal{S}_{\mathcal{A}+B}^{P_{\mathfrak{A}}, q} \subseteq \left( \mathcal{S}_{\mathcal{A}+B}^{P_{\mathfrak{A}}, q} \right)^{\perp_e}$ (respectively, $ \mathcal{S}_{\mathcal{A}+B}^{P_{\mathfrak{A}}, q^2} \subseteq \left( \mathcal{S}_{\mathcal{A}+B}^{P_{\mathfrak{A}}, q^2} \right)^{\perp_h}$) if for each cyclotomic coset $\Lambda_e^q$ (respectively, $\Lambda_e^{q^2}$) in $\mathcal{A}+B$, one has that $\Lambda_e^q \cap (\mathbb{Z}_N \setminus -(\mathcal{A}+B)_N ) \neq \emptyset$ (respectively, $\Lambda_e^{q^2} \cap (\mathbb{Z}_N \setminus -(\mathcal{A}+B)_N^q ) \neq \emptyset$).

Whether one considers H-BCH codes or subfield-subcodes of $\{1\}$-affine variety codes, by Equality (5) in \cite{HBCH}, one has self-orthogonality when considering exponents in  $\mathcal{A}+B$ if one assumes that the code $\mathcal{S}_{B}^{P_{n}, q}$ (respectively, $\mathcal{S}_{B}^{P_{n}, q^2}$) is self-orthogonal.


Let $\lambda \leq \frac{N}{n}$ be a positive integer and set
$$R:=\left\{1,\zeta_{n},\dots,\zeta_{n}^{n-1},a,a\zeta_{n},
\dots,a\zeta_{n}^{n-1},\dots,a^{\lambda-1},a^{\lambda-1}\zeta_{n},\dots,a^{\lambda-1}\zeta_{n}^{n-1}\right\}.$$
When evaluating at  $R$, we get H-BCH codes if $\lambda n$ does not divide $N$ and subfield-subcodes of $\{1\}$-affine variety codes otherwise. From Theorems \ref{eval}, \ref{SingletonQ} and  \ref{mats-2}, one deduces the following result.

\begin{teo}
\label{qlrcbch}
Consider a $q$-complete (respectively, $q^2$-complete if $s=2\varsigma$) modulo $n$ set $B$, such that it contains $d-1$ consecutive positive integers and $d \leq n$. Define $\Delta=  \mathcal{A} + B $ and suppose that $\delta$ is a positive integer such that $2 \leq \delta \leq d $. Finally, assume that the code $\mathcal{S}_{B}^{P_n, q}$ (respectively, $\mathcal{S}_{B}^{P_n, q^2}$) is Euclidean (respectively, Hermitian) self-orthogonal. Then:
\begin{enumerate}
\item The code $\mathcal{S}_{\Delta}^{R, q}$ is Euclidean self-orthogonal and it gives rise to a quantum $(r, \delta)$-locally recoverable code with parameters $[[\lambda n,  \lambda (n- 2 \#B), \geq d]]_q$ and locality $(n - \delta +1, \delta)$.
\item If $s = 2 \varsigma$, the code $\mathcal{S}_{\Delta}^{R, q^2}$ is Hermitian self-orthogonal and it gives rise to a quantum $(r, \delta)$-locally recoverable code with parameters $[[\lambda n,  \lambda (n- 2 \#B), \geq d]]_q$ and locality $(n - \delta +1, \delta)$.
\end{enumerate}
\end{teo}

\section{Some examples of optimal pure quantum \texorpdfstring{$(r,\delta)$}{lr5}-LRCs}
\label{LA4}
In this section, we use the previous results to construct some examples of  optimal pure quantum $(r,\delta)$-LRCs.  

\subsection{Optimal pure \texorpdfstring{$q$}{q4}-ary quantum \texorpdfstring{$(r,\delta)$}{lr6}-LRCs. Case \texorpdfstring{$n$ divides $q^2+1$}{n2}}
\label{41}
In this subsection we suppose that $n$ divides $q^2+1$ which divides $N$ and it divides $q^s -1$ for some $s=2 \varsigma$. We also assume that
$n$ is odd and $n \geq 3 $.  We will provide optimal pure $q$-ary quantum $(r,\delta)$-LRCs of length $\lambda n$, where $\lambda \leq \frac{N}{n}$. The following three results will be key for determining the parameters of our quantum codes.

First we define the following sets of integer numbers which depend on a positive integer $u$, $1 \leq u \leq \frac{n-1}{2}$:
\begin{equation}
\label{BU}
B(u) := \left\{ (n-1)/2 + i \;:\; 1 \leq i \leq u \right\} \cup \left\{ 1+ (n-1)/2 - i \;:\; 1 \leq i \leq u \right\}.
\end{equation}
Our first result in this subsection is a consequence of Corollary \ref{mats-1} and Theorems \ref{eval} and \ref{mats-2}. Recall that $R$ is the set of points where we evaluate.
\begin{pro}
\label{3-jap}
With the above notation, set $\Delta(u) = \mathcal{A} + B(u)$, $1 \leq u \leq \frac{n-1}{2}$. Then, the code $\left( \mathcal{S}_{\Delta(u)}^{R, q^2} \right)^{\perp_h}$ is an optimal $(r=n - 2u , \delta = 2u+1)$-locally recoverable code  with parameters $[\lambda n, \lambda n - 2 \lambda u , 2u+1]_{q^2}$.
\end{pro}
\begin{proof}
We start by proving that $B(u)$ is $q^2$-complete modulo $n$. Indeed, for $1 \leq \ell \leq n-1$, the following equivalence holds:
\begin{equation}
  \ell \in B(u) \Leftrightarrow n-\ell \in B(u), \label{eq:m1}
\end{equation}
and, since $n$ divides $q^2+1$, $\ell q^2 \equiv n- \ell \mod{n}$, which shows that $B(u)$ is $q^2$-complete modulo $n$ by (\ref{eq:m1}). Then, by Proposition \ref{1-2-jap}, the set $\Delta (u)$ is $q^2$-complete modulo $N$.
Thus, Corollary \ref{mats-1}
proves that the code $\left( \mathcal{S}_{\Delta(u)}^{R, q^2} \right)^{\perp_h}$ is a classical $(r, \delta)$-LRC. 

Now,
independently of whether $\lambda n$ divides $N$ or not, by Theorems \ref{eval} and \ref{mats-2}, the left hand side of the Singleton-like bound (\ref{Singleform}) equals
\begin{eqnarray}
  &&(\lambda n-2u \lambda) + (1+2u) + \left(\left\lceil \frac{\lambda n-2u \lambda }{n-2u}\right\rceil-1\right)2u \label{eq:m10}\\
  &=& \lambda(n-2u) + (1+2u) + (\lambda -1) 2u\nonumber\\
  &=&\lambda n+1\nonumber
\end{eqnarray}
and we get an optimal $(r=n - 2u , \delta=2u+1)$-locally recoverable code  with parameters $[\lambda n, \lambda n - 2 \lambda u , 2u+1]_{q^2}$.
\end{proof}

Let us state a second result which will be useful.

\begin{pro}
\label{sumsquare}
Let $q$ be a prime power and $n$ an odd integer such that $2 \leq n \leq q^2+1$ and it divides $q^2+1$. Then, there exists a unique pair of positive integers $(m_1,m_2)$ such that $n=m_1^2 + m_2^2$ and $m_1 \equiv m_2 q \mod n$.
\end{pro}

\begin{proof}
Denote by $\mathbb{Z}[i]$ the ring of Gaussian integers and consider the ideal $\mathfrak{a} \subseteq \mathbb{Z}[i]$ defined as follows:
\[ \mathfrak{a} = \{ x + yi \in \mathbb{Z}[i] \; : \; x \equiv qy \mod n \}. \]
To verify that $\mathfrak{a}$ is an ideal, it suffices to note that $i(x+yi)= -y +xi$ and $-y +xi \in \mathfrak{a}$ since $x \equiv qy \mod n$ implies $qx \equiv q^2y \equiv - y \mod n$ because $q^2 \equiv -1 \mod n$.

It is clear that $\mathfrak{a}$ is generated by $n$ and $q+i$ and since $\mathbb{Z}[i]$ is a principal ideal domain, there is an element $\mu = m_1 + m_2 i$ generating $\mathfrak{a}$. By definition of $\mathfrak{a}$, $m_1 \equiv m_2 q \mod n$. Now, the size of the quotient ring $\frac{\mathbb{Z}[i]}{\mathfrak{a}}$ (usually called norm of $\mathfrak{a}$) is $n$ because $q+i$ equals zero in $\frac{\mathbb{Z}[i]}{\mathfrak{a}}$ and, thus, the classes of $0, 1, \ldots, n-1$ generate $\frac{\mathbb{Z}[i]}{\mathfrak{a}}$.

The fact that $\mu$ generates $\mathfrak{a}$ proves that $n=m_1^2 + m_2^2$ and it remains to show that $m_1$ and $m_2$ can be chosen as positive integers. For a start $n \geq 2$, and thus $m_1=0$ and $m_2=0$ is not possible. The equality $m_1 = 0$ cannot hold; indeed, in this case $n = m_2^2$ and the equality $0 \equiv m_2 q \mod n$ implies $m_2^2$ divides $m_2 q$, which allows us to deduce that $m_2$ divides $q$. Thus $m_2^2 = n$ divides $q^2$, but $n$ divides $q^2+1$, which means that $n$ divides both $q^2$ and $q^2+1$, a contradiction since $n \geq 2$.

In conclusion, both $m_1$ and $m_2$ are different from zero and since $\mathfrak{a}$ has only four valid generators:
\[ (m_1, m_2), \quad (-m_1, -m_2), \quad (-m_2, m_1), \quad (m_2, -m_1), \]
the result is proved because there is only one choice with positive coordinates.
\end{proof}

From now, and until the end of this subsection, according to the above proposition, we  consider a positive integer $n$ as above, $n= m_1^2+m_2^2$, where $m_1$ and $m_2$ are both positive integers and $m_1 \equiv  m_2 q \mod n$.

Our third result is the following lemma. It helps to prove the forthcoming Theorem \ref{4-5-japo}, which is the main result in this subsection. This last result shows a way of getting optimal pure $q$-ary quantum $(r,\delta)$-LRCs.

\begin{lem}\label{lem:selforthogonal}
Let $n$, $q$, $m_1$ and $m_2$  as in Proposition \ref{sumsquare}. Assume that $n$ is odd, then the classical code $\mathcal{S}_{\Delta(u)}^{R, q^2}$
is Hermitian self-orthogonal
if and only if $u \leq (m_1 + m_2)/2$.
\end{lem}

\begin{proof}
Keep the  above notation and consider the (related to $B(u)$) set:
\[
C(u) := \left\{ c \in \mathbb{Z}_{\geq 0} \;: \; c \equiv -qb \mod n \mbox{ for some } b \in B(u) \right\}.
\]

In a first step we are going to prove the following formula.
\begin{equation}
\label{BBBBB}
\# \left( B(u) \cap C(u) \right) = \# \big\{ (x,y) \in \mathbb{Z}^2 \; : \; |x m_1 + y m_2 | \leq 2u-1 \; \mbox{ and } \; |x m_2 - y m_1 | \leq 2u-1 \big\}.
\end{equation}
Let $c$ be an element in  $B(u) \cap C(u)$, it means that $c \in C(u)$ and there exists $b \in B(u)$ such that $c \equiv -qb \mod n$. Consider integer numbers $v = 2b - n$ and $w = 2c - n$. Both belong to $B(u)$, which implies $|v| \leq 2u-1$ and $|w| \leq 2u-1$. In addition $c \equiv -qb \mod n$ and then $w \equiv -qv \mod n$. Thus,
\begin{equation}
\label{AAAAA}
\# \left( B(u) \cap C(u) \right) = \# \big\{ (v,w) \in \mathbb{Z}^2 \; : \; v, w \; \mbox{ are odd and } \; |v| \leq 2u-1 \; \mbox{ and } \; |w| \leq 2u-1 \big\}.
\end{equation}

Let $T: \mathbb{Z}^2 \rightarrow \mathbb{Z}^2$ be the mapping of integer pairs $$T(x,y)=(v'=x m_1 + y m_2, w'=x m_2 - y m_1).$$ Recalling that $m_1 \equiv qm_2 \mod n$, one has $ v' \equiv m_2(qx + y) \mod n$, and
multiplying by $q$, one gets
$qv' \equiv  q^2m_2x + qm_2y \mod n$. Finally, the fact that $q^2 \equiv -1 \mod n$ proves
\[ qv' \equiv - m_2x + (qm_2)y \equiv - m_2x + m_1y \equiv -w' \mod n. \]

We claim that $T$ is a bijection between $\mathbb{Z}^2$ and the sublattice
\[
L := \{(v',w') \in \mathbb{Z}^2 \mid w' \equiv -qv' \mod n\}.
\]
To establish this, recall we have already shown that the image of $\mathbb{Z}^2$ under $T$, denoted $T(\mathbb{Z}^2)$, is contained in $L$. Thus, it suffices to show that $T(\mathbb{Z}^2) = L$. Since the determinant of the matrix corresponding to $T$ is $-n$, the index of $T(\mathbb{Z}^2)$ in the full lattice $\mathbb{Z}^2$ is given by:
$
[\mathbb{Z}^2 : T(\mathbb{Z}^2)] = |-n| = n$.  To prove the equality of the sublattices, it is now enough to show that $[\mathbb{Z}^2 : L] = n$. Consider the group homomorphism:
\[
\varphi: \mathbb{Z}^2 \rightarrow \mathbb{Z}/n\mathbb{Z}, \quad \varphi(v',w') = (w'+qv') \bmod n.
\]
By definition, the kernel of $\varphi$ is precisely $L$. Since $\varphi$ is surjective, the first isomorphism theorem shows $
[\mathbb{Z}^2 : L] = n = [\mathbb{Z}^2 : T(\mathbb{Z}^2)]$.  Since $T(\mathbb{Z}^2) \subseteq L$ and both sublattices have the same finite index in $\mathbb{Z}^2$, we conclude that $T(\mathbb{Z}^2) = L$, completing the proof of our claim.

Continuing with the proof, we know that  $n = m_1^2 + m_2^2$ is odd, therefore one of $m_1$ and $m_2$ is odd and the other one is even, which allows us to deduce that $v'$ and $w'$ are both odd integers if and only if both $x$ and $y$ are odd integers. Finally, considering the subset of $L$ given by the pairs $(v,w)$ before considered, the constraints $|v| \leq 2u-1$ and $|w| \leq 2u-1$ in (\ref{AAAAA}) precisely translate to $|x m_1 + y m_2| \leq 2u-1$ and $|x m_2 - y m_1| \leq 2u-1$ and, therefore Equality (\ref{BBBBB}) is proved.\\

To finish the proof of the lemma we have only to prove the following equivalence:
\begin{equation}
\label{CCCCC}
 B(u) \cap C(u) = \emptyset \; \mbox{ if and only if } \; 2u \leq m_1 + m_2.
\end{equation}

In fact, (\ref{CCCCC}) is true if and only if there are no odd pairs $(x,y)$ satisfying the constraints in (\ref{AAAAA}).

Set $\mathfrak{V}(x,y) = |x m_1 + y m_2|$ and $\mathfrak{W}(x,y) = |x m_2 - y m_1|$, and consider the map $h: \mathbb{Z}^2 \rightarrow \mathbb{Z}$, $h(x,y):=\max \{\mathfrak{V}(x,y), \mathfrak{W}(x,y)\}$. We are interested in the absolute minimum:
\[
\mathfrak{m} := \min\{ h(x,y) \; : \; x, y \; \mbox{ are odd integer numbers} \}.
\]
Let us look for that minimum. For a start, it is straightforward to show that  $\mathfrak{V}(x,y)^2 + \mathfrak{W}(x,y)^2 = n (x^2+y^2)$. Now, starting with the set of odd integers of minimum magnitude $\{-1,1\}$, by direct computation, it holds that $\min \{ h(x,y) : x, y \in \{-1,1\} \} =m_1+m_2$,
because $m_1, m_2 >0$ and then $m_1 + m_2 > |m_1-m_2|$.

Now, when either $x$ or $y$ does not belong to $\{-1,1\}$, one gets  $x^2+y^2 \geq 3^2+1^2 = 10$ and then $\mathfrak{V}(x,y)^2 + \mathfrak{W}(x,y)^2 \geq 10 n$, which proves that $h(x,y)^2 \geq 5n$. Finally, the sequence of equalities $$(m_1+m_2)^2 = m_1^2 + m_2^2 + 2m_1m_2 \leq 2(m_1^2+m_2^2) \leq 2n < 5n$$ allows us to deduce that $h(x,y) \geq \sqrt{5n} > m_1+ m_2$. Thus, by (\ref{BBBBB}), $B(u) \cap C(u) = \emptyset$ if and only if $2u-1 < \mathfrak{m}$, which is equivalent to $2u < m_1 +m_2$ and the lemma is proved.

\end{proof}

\begin{teo}
\label{4-5-japo}
Keep the above notation, where $n$ is an odd divisor of $q^2+1$ (which divides $N$ and it divides $q^s -1$) and $m_1$ and $m_2$ is the unique pair of positive integers
such that $m_1^2 + m_2^2 = n$ and $m_1 \equiv m_2 q$ modulo $n$. Suppose also that $\lambda \leq \frac{N}{n}$.
Then, the code $\left( \mathcal{S}_{\Delta(u)}^{R, q^2} \right)^{\perp_h}$, $1 \leq u \leq \frac{{m_1 + m_2 }}{2} \leq \frac{n-1}{2}$, is Hermitian dual-containing, giving rise to an optimal pure quantum $(r= n - 2u , \delta= 2u+1)$-LRC with parameters $[[\lambda n, \lambda n - 4 \lambda u , 2u+1]]_{q}$.
\end{teo}
\begin{proof}
The Hermitian dual-containment follows from the last lemma.
To prove that the obtained quantum code is pure, we notice that  $\delta = 2u+1$ is the minimum distance of $\left(\mathcal{S}_{\Delta(u)}^{R, q^2}\right)^{\perp_h}$, while
\[
\dis \left(\mathcal{S}_{\Delta(u)}^{R, q^2}\right) \geq \lambda n - \left( \frac{n-1}{2} +u \right)> \delta,
\]
where the first inequality holds because a polynomial of degree $\frac{n-1}{2}+ u$ has, at most, $\frac{n-1}{2}+ u$ zeroes.
\end{proof}


\subsection{Optimal pure \texorpdfstring{$q$}{Q1}-ary quantum \texorpdfstring{$(r,\delta)$}{lr7}-LRCs. Case \texorpdfstring{$n=\mathfrak{q}-1$}{n}}
\label{42}
In this subsection, the symbol $\mathfrak{q}$ stands for $q$ when using the Euclidean construction and for $q^2$ in the Hermitian case. Set $n=\mathfrak{q} -1$ and  $s=2\varsigma$ when considering Hermitian duality. As before, $\lambda \leq \frac{N}{n}$, where $n$ divides $N$ and it divides $q^s -1$.  For any positive integer $v < \mathfrak{q} -1$, set
\[
B'(v):=\{1, \ldots, v\}.
\]
$B'(v)$ is $\mathfrak{q}$-complete modulo $n$ because $\mathfrak{q} \equiv 1 \mod n$.

Let $\Delta'(v) = \mathcal{A} + B'(v)$. Reasoning as in Proposition \ref{3-jap}, but considering $\mathfrak{q}$ instead of $q^2$, one gets that $\left( \mathcal{S}_{\Delta'(v)}^{R, q} \right)^{\perp_e}$ and $\left( \mathcal{S}_{\Delta'(v)}^{R, q} \right)^{\perp_h}$  are $\mathfrak{q}$-ary $(n-v,v+1)$-LRCs. Let us show that they are optimal, which follows from 
the following sequence of equalities:
\begin{eqnarray*}
  &&(\lambda n - \lambda v ) + (1+v) + \left(\left\lceil \frac{\lambda n - \lambda v}{n-v}\right\rceil-1\right)v \\
  &=& \lambda (n-v) + (1+v) + \left( \lambda-1\right)v\\
  &=&\lambda n +1.
\end{eqnarray*}

In addition, if $\mathfrak{q}=q$, it suffices that $2v \leq n-1$ to guarantee $B'(v) \cap (-B'(v)) = \emptyset$ and, when $\mathfrak{q}=q^2$,  $(1+q)v \leq n-1$ implies that $B'(v) \cap (-B'(v)^q) = \emptyset$. Therefore, we have proved the following result.

\begin{teo}
\label{jap-s-3}
Keep the above notation and  recall that $\Delta'(v) = \mathcal{A} + B'(v)$.
\begin{enumerate}
\item Set $n=q-1$, $q \geq 3$ and suppose that $1 \leq v \leq \frac{n-1}{2}$. Then, the code $\left( \mathcal{S}_{\Delta'(v)}^{R, q} \right)^{\perp_e}$ is an optimal $(r=n - v , \delta = v+1)$-locally recoverable code  with parameters $[\lambda n, \lambda n -  \lambda v, v+1]_{q}$. It is Euclidean dual-containing giving rise to an optimal pure quantum $(r= n - v , \delta= v+1)$-LRC with parameters $[[\lambda n, \lambda n - 2 \lambda v , v+1]]_{q}$.

\item Set $n=q^2-1$ and suppose that $1 \leq v \leq \frac{n-1}{q+1}$. Then, the code $\left( \mathcal{S}_{\Delta'(v)}^{R, q^2} \right)^{\perp_h}$ is an optimal $(r=n - v , \delta=v+1)$-locally recoverable code  with parameters $[\lambda n, \lambda n -  \lambda v, v+1]_{q^2}$. It is Hermitian dual-containing giving rise to an optimal pure  quantum $(r= n - v , \delta= v+1)$-LRC with parameters $[[\lambda n, \lambda n - 2 \lambda v , v+1]]_{q}$.

\end{enumerate}

\end{teo}

When $q$ is odd, Item (2) of the above theorem can be improved when $\lambda=2$. The specific result is the following one, where Hermitian dual containment holds by considering the case $s=2$ and $n_1= 2(q-1)$  in \cite[Theorem A]{HBCH}.

\begin{cor}
\label{BCHER}
Let $n=q^2-1$, $q$ odd and $\lambda =2$ and suppose that $1 \leq v \leq 2q-3$ . Then, the code $\left( \mathcal{S}_{\Delta'(v)}^{R, q^2} \right)^{\perp_h}$ is an optimal $(r=n - v , \delta=v+1)$-locally recoverable code  with parameters $[2 n, 2 n -  2 v, v+1]_{q^2}$. It is Hermitian dual-containing giving rise to an optimal pure quantum $(r= n - v , \delta= v+1)$-LRC with parameters $[[2 n, 2 n - 4  v , v+1]]_{q}$.
\end{cor}


\subsection{Optimal pure \texorpdfstring{$q$}{q2}-ary quantum \texorpdfstring{$(r,\delta)$}{lr9}-LRCs from evaluation of multivariate polynomials}
\label{43}
In this subsection, we will use monomial-Cartesian codes and results in Subsections \ref{41} and \ref{42} to provide new optimal pure $q$-ary quantum $(r,\delta)$-LRCs.

We begin by recalling the concept of {\it monomial-Cartesian code}. We follow the definition given in \cite[Definition 3.1]{GFMC} and, since we are interested in subfield-subcodes over $\mathbb{F}_{q}$, we need a larger field $\mathbb{F} =\mathbb{F}_{q^s}$. Some previous references are \cite{G2008, LMS2020}. Consider a Cartesian product $Z:=Z_1 \times \cdots \times Z_w = \{\mathbf{z}_1, \ldots, \mathbf{z}_\beta\}$, where $Z_\ell$, $1 \leq \ell \leq w$, are subsets of $\mathbb{F}$ of cardinality $n_\ell \geq 2$. Denote by $\mathbb{F}[X_1, \ldots, X_w]$ the polynomial ring in $w$ variables with coefficients in $\mathbb{F}$ and, for $1 \leq \ell \leq w$, take $f_\ell (X_\ell) = \prod_{z \in Z_\ell} (X_\ell -z) \in \mathbb{F}[X_1, \ldots, X_w]$. Let $I$ be the ideal of $\mathbb{F}[X_1, \ldots, X_w]$ generated by $\{f_1, \ldots, f_w\}$ and denote by $\mathcal{R}$ the quotient ring $\mathcal{R} := \mathbb{F}[X_1, \ldots, X_w]/I$. It is clear that any class $f$ in $\mathcal{R}$ has a representative as follows:
\[
f= \sum_{ \mathbf{e}=(e_1, \ldots, e_w)\in E}  f_{\mathbf{e}} X_1^{e_1} \cdots X_w^{e_w},
\]
where $E:= \{0, \ldots, n_1-1\} \times \cdots \times \{0, \ldots, n_w-1\}$. The evaluation map
\[
\mathrm{ev}^Z := \mathcal{R} \rightarrow \mathbb{F}^\beta \; \mbox{ defined by } \; \mathrm{ev}^Z(f)= (f(\mathbf{z}_1), \ldots, f(\mathbf{z}_\beta))
\]
allows us to define the {\it monomial-Cartesian code} $\mathcal{C}_\Gamma^Z$ given by a non-empty set $\Gamma \subseteq E$ as
\[
\mathcal{C}_\Gamma^Z := \langle \mathrm{ev}^Z(X^\mathbf{e}) \; : \; \mathbf{e} \in \Gamma \rangle,
\]
where $\langle \cdot \rangle$ means the $\mathbb{F}$-linear space spanned by $\cdot$ and $X^\mathbf{e}=X_1^{e_1} \cdots X_w^{e_w}$.

Next we introduce the particular class of monomial-Cartesian codes we are interested in. We will show that Theorems \ref{4-5-japo}, \ref{jap-s-3} and Corollary  \ref{BCHER} can be extended to larger lengths. According to the above results, in the rest of this subsection we consider two situations:

{\it Situation 1}, where $Z_1$ and the first projection of $\Gamma$ are as in Subsection \ref{41}, and

{\it Situation 2}, where the same happens with respect to Subsection \ref{42}.

Denote by $\mathfrak{n}$ the value $n$ considered in Situation 1 and $\mathfrak{n} =\mathfrak{q}-1$ in Situation 2.  Recall that, when we use Hermitian duality $\mathfrak{q} = q^2$ and $\mathfrak{q}$ is $q$ otherwise.

The following lemma considers the above notation, in particular that in Subsection \ref{22}. It will be shown that, when considering certain homothetic evaluation codes, one can replace the set $\mathcal{A}$ with a truncated set
\[
\mathcal{A}_\lambda := \{ i\, n \mid 0 \le i \le \lambda - 1 \}.
\]

\begin{lem}
\label{lemauno}
Let $B$ be a subset of $\mathbb{Z}_n$, then the following equality of homothetic evaluation codes holds:
\[
\mathcal{H}_{\mathcal{A} + B}^P = \mathcal{H}_{\mathcal{A}_\lambda + B}^P.
\]
\end{lem}
\begin{proof}
Recall that $\lambda < \frac{N}{n}$ such that the block length $\lambda n$ does not divide $N$. $P$ can be partitioned as
\[
P = \bigcup_{j=0}^{\lambda-1} P_j, \quad \text{where } P_j = \{a^j \zeta_n^k \mid 0 \le k \le n-1\}.
\]
Let $\pi_{P_j}: \mathbb{F}_{q^s}^{\lambda n} \to \mathbb{F}_{q^s}^n$ denote the projection map onto the block $P_j$.
Let $0 \le i \le \frac{N}{n}-1$ and $b \in B$, then
\[
\pi_{P_j}\left(\mathrm{ev}^P(X^{in+b})\right) = \left( (a^j \zeta_n^k)^{in+b} \right)_{k=0}^{n-1} = \omega^{ji} a^{jb} \cdot \mathrm{ev}^{P_0}(X^b),
\]
where $\omega = a^n$ is a primitive $\frac{N}{n}$-th root of unity.
This shows that
\[
\pi_{P_j}\left(\mathcal{H}_{\mathcal{A}+B}^P\right) = \pi_{P_j}\left(\mathcal{H}_{\mathcal{A}_\lambda+B}^P\right) = a^{jb} \cdot \mathcal{H}_B^{P_0}
\]
and thus $\dim_{\mathbb{F}_{q^s}}\left(\pi_{P_j}(\mathcal{H}_{\mathcal{A}+B}^P)\right) = \#B$. Therefore,
\begin{equation}
\label{RYU1}
\dim_{\mathbb{F}_{q^s}}\left(\mathcal{H}_{\mathcal{A}+B}^P\right) \le \sum_{j=0}^{\lambda-1} \dim_{\mathbb{F}_{q^s}}\left(\pi_{P_j}(\mathcal{H}_{\mathcal{A}+B}^P)\right) = \lambda \#B.
\end{equation}
Now,
\[
\dim_{\mathbb{F}_{q^s}}\left(\mathcal{H}_{\mathcal{A}_\lambda+B}^P\right) = \lambda \#B
\]
because of the non-singularity of the Vandermonde matrix $M = (\omega^{j i})_{0 \le j, i \le \lambda-1}$ corresponding to $\mathcal{H}_{\mathcal{A}_\lambda+B}^P$.

Finally, the fact that $\mathcal{H}_{\mathcal{A}_\lambda+B}^P \subseteq \mathcal{H}_{\mathcal{A}+B}^P$ along with \eqref{RYU1} forces the equality of dimensions and spaces, which concludes the proof.
\end{proof}

Now, let $R$ be as in Subsections \ref{41} and \ref{42} and define $Z_1 := R$, then $\# Z_1 = \lambda \mathfrak{n}$ for positive integers $\mathfrak{n}$ and $\lambda$ such that $\lambda \leq \frac{N}{n}$. Take positive integers $2 \leq n_\ell < \mathfrak{q}$, $2 \leq \ell \leq w$,  and consider subsets $Z_\ell$, $2 \leq \ell \leq w$, in $\mathbb{F}_{\mathfrak{q}}$; that is, the sets $Z_\ell$ (with the exception of $Z_1$) have no elements in $\mathbb{F} \setminus \mathbb{F}_{\mathfrak{q}}$.

From this point through the end of the subsection, let $\Delta$ denote the set $\mathcal{A}_\lambda+B$, where $B$ stands for  the set $B(u)$ in Situation 1 and $B'(v)$ in Situation 2. Note that by Lemma \ref{lemauno}, one can consider $\mathcal{A}_\lambda$ instead of $\mathcal{A}$, which remains valid even when $\lambda n$ divides $N$ since, in this case, we can set $N = \lambda n$.

In both situations, we set
\[
\Gamma = \Delta \times \{0, 1, \dots, n_2 -1\} \times \cdots \times \{0, 1, \dots, n_w -1\}.
\]
As before,
\[
Z = Z_1 \times \cdots \times Z_w \quad \text{and} \quad E = \{0, 1, \dots, \lambda \mathfrak{n} -1\} \times \{0, 1, \dots, n_2-1\} \times \cdots \times \{0, 1, \dots, n_w -1\}.
\]
Previously, we proved that $\mathcal{C}_\Delta^{Z_1}$ is self-orthogonal with respect to the Hermitian dual and, in the second setting, the Euclidean dual. Consequently, we have
\[
\mathcal{C}^{Z_1}_{\Delta^\perp} = \left( \mathcal{C}_\Delta^{Z_1} \right)^\perp \supseteq \mathcal{C}_\Delta^{Z_1},
\]
where $\perp$ denotes, depending on the case, either the Hermitian dual $\perp_h$ or the Euclidean dual $\perp_e$ and

\begin{align}
  \label{RYU2-1}
  \Delta^\perp &:= \mathcal{A}_\lambda + \left(\{0, 1, \ldots, n-1\} \setminus -B(u)^q\right) \quad \text{in Situation 1}, \\
  \label{RYU2-2}
  \Delta^\perp &:= \mathcal{A}_\lambda + \left(\{0, 1, \ldots, n-1\} \setminus -B'(v)\right) \quad \text{in Situation 2}.
\end{align}

Note that when the second equality corresponds to the Hermitian case, $-B'(v)$ means $-B'(v)^q$.

Definitions (\ref{RYU2-1}) and (\ref{RYU2-2}) of $\Delta^\perp$ are motivated by the fact that the Hermitian or Euclidean inner product of the evaluations of the monomials $X_1^{\kappa n+\iota}$ and $X_1^{\xi n-\iota}$, where $0 \le \kappa, \xi \le \lambda - 1$ and $\iota \in B$, does not vanish. This highlights the exact exponents that must be eliminated when computing the dual. Therefore, to obtain the dual, by dimensionality arguments, we must consider the exponents indicated in \eqref{RYU2-1} and \eqref{RYU2-2}.

At this stage, we set
\[
\Gamma^\perp = \Delta^\perp \times \{0, 1, \dots, n_2 - 1\} \times \cdots \times \{0, 1, \dots, n_w - 1\}.
\]
Consequently, we obtain $\left(\mathcal{C}^{Z}_{\Gamma}\right)^\perp = \mathcal{C}^Z_{\Gamma^\perp}$. This identity holds because $\# \Delta + \#\Delta^\perp = \lambda n$ and, for $\mathbf{e} \in \Gamma$, $e'_1 \in \Delta^\perp$, and $e'_i \in \{0, 1, \ldots, n_i - 1\}$, $2 \leq i \leq w$, the (Hermitian or Euclidean) inner product
\[
\mathrm{ev}^Z (X^\mathbf{e}) \cdot \mathrm{ev}^Z (X_1^{e'_1} X_2^{e'_2} \cdots X_w^{e'_w})
\]
vanishes whenever one of its constituent factors (involving sums of powers of roots) is zero. This is guaranteed since $\mathrm{ev}^{Z_1} \left(X_1^{e_1}\right) \cdot \mathrm{ev}^{Z_1} (X_1^{e'_1}) = 0$ for $e_1 \in \Delta$ and $e'_1 \in \Delta^\perp$.

According to \cite[Proposition 3.4]{GFMC}, the length of $\mathcal{C}_\Gamma^Z$ is $\lambda \mathfrak{n}n_2 \cdots n_w$, and its dimension is $2\lambda un_2 \dots n_w$ in Situation 1, and $\lambda v n_2 \dots n_w$ in Situation 2. Furthermore, the dimension of $\mathcal{C}^Z_{\Gamma^\perp}$ is as expected.

Let $\mathcal{S}_\Gamma^{Z,\mathfrak{q}}$ denote the subfield-subcode $\mathcal{C}_\Gamma^Z \cap \mathbb{F}_{\mathfrak{q}}^{\lambda \mathfrak{n}n_2 \dots n_w}$. The fact that $z_\ell \in \mathbb{F}_{\mathfrak{q}}$ for $2 \leq \ell \leq w$ implies that $\mathcal{S}_\Gamma^{Z,\mathfrak{q}}$ is uniquely determined by $\mathcal{S}_\Delta^{R,\mathfrak{q}}$.

Now, following the approach of \cite[Theorem 4 and Section 4]{QINP} for the Euclidean case and \cite[proof of Theorem 7]{QINP} for the Hermitian case, we obtain:
\[
\left( \mathcal{S}_\Gamma^{Z,\mathfrak{q}} \right)^{\perp} = \mathbf{tr} \left[ \left(\mathcal{C}_\Gamma^Z\right)^{\perp} \right] = \mathbf{tr} \left(\mathcal{C}_{\Gamma^{\perp}}^Z\right) = \mathcal{S}_{\Gamma^{\perp}}^{Z,\mathfrak{q}},
\]
where $\mathbf{tr}$ denotes the trace map applied componentwise, and $\mathcal{S}_{\Gamma^{\perp}}^{Z,\mathfrak{q}}$ represents the subfield-subcode of $\mathcal{C}_{\Gamma^{\perp}}^Z$ over $\mathbb{F}_{\mathfrak{q}}$.

The code $\mathcal{C}_{\Gamma^{\perp_h}}^Z$ has length $\lambda \mathfrak{n}n_2 \cdots n_w$ and dimension $\lambda \mathfrak{n}n_2 \cdots n_w - 2\lambda u n_2 \cdots n_w$ in Situation 1, and $\mathcal{C}_{\Gamma^{\perp_e}}^Z$ and
$\mathcal{C}_{\Gamma^{\perp_h}}^Z$ have  length $\lambda \mathfrak{n}n_2 \cdots n_w$ and dimension $\lambda \mathfrak{n}n_2 \cdots n_w - \lambda v n_2 \cdots n_w$ in Situation 2.

With respect to distances, setting $\perp$ as above, one can apply \cite{LMS2020} (see also \cite{Soprunov}) to deduce that
\[
\dis \left( \mathcal{C}_{\Gamma^{\perp}}^Z \right) \geq \dis \left( \mathcal{C}^{Z_1}_{\Delta^\perp} \right) \dis \left( \mathcal{C}_{\{0, 1, \ldots, n_2 -1\}}^{Z_2} \right) \cdots \dis \left( \mathcal{C}_{\{0, 1, \ldots, n_w -1\}}^{Z_w} \right)
\]
and, as a consequence, $\dis \left( \mathcal{C}_{\Gamma^{\perp}}^Z \right) \geq 2u+1$ in Situation 1, and $\dis \left( \mathcal{C}_{\Gamma^{\perp}}^Z \right) \geq v+1$ in Situation 2. Therefore,
we deduce that the parameters of $\left( \mathcal{S}_\Gamma^{Z,\mathfrak{q}} \right)^{\perp_h}$ and $\left( \mathcal{S}_\Gamma^{Z,\mathfrak{q}} \right)^{\perp_e}$ in the previously studied situations are
$$[\lambda \mathfrak{n}n_2 \cdots n_w, \lambda \mathfrak{n}n_2 \cdots n_w - \lambda 2u n_2 \cdots n_w, 2u+1]_{\mathfrak{q}}$$
and
$$[\lambda \mathfrak{n}n_2 \cdots n_w, \lambda \mathfrak{n}n_2 \cdots n_w - \lambda v n_2 \cdots n_w, v+1]_{\mathfrak{q}}.$$

Finally, the fact that $Z_\ell \subseteq \mathbb{F}_{\mathfrak{q}}$, $2 \leq \ell\leq w$, shows that $\left( \mathcal{S}_\Gamma^{Z,\mathfrak{q}} \right)^{\perp_h}$ and $\left( \mathcal{S}_\Gamma^{Z,\mathfrak{q}} \right)^{\perp_e}$ inherit the dual containment property and $(r,\delta)$-locality of $\left( \mathcal{S}_\Delta^{P,\mathfrak{q}} \right)^{\perp_h}$ and $\left( \mathcal{S}_\Delta^{P,\mathfrak{q}} \right)^{\perp_e}$. Therefore, Theorems \ref{4-5-japo} and \ref{jap-s-3}, and Corollary  \ref{BCHER} prove the next result, where we recall that $n$ divides $N$ and it divides $q^s -1$.

\begin{teo}
\label{teomonomial}
The following statements hold:
\begin{description}
\item[A] Suppose that $s=2\varsigma$ and $n \geq 3$ is odd and divides $q^2 +1$. Take positive integers $n_\ell$, $2 \leq n_\ell < q^2, 2 \leq \ell \leq w$ and  $u \leq \frac{m_1 +m_2}{2}$, for values $m_1$ and $m_2$ as in Proposition \ref{sumsquare}. Then,  there exists an optimal pure quantum $(r=n-2u, \delta=2u+1)$-locally recoverable code with parameters
$$[[\lambda n n_2 \cdots n_w, \lambda n n_2 \cdots n_w - 4 \lambda u n_2 \cdots n_w, 2u+1]]_{q}.$$
\item[B] Set $n=q-1$.  Take positive integers $n_\ell$, $2 \leq n_\ell < q, 2 \leq \ell \leq w$ and $v \leq \frac{q-2}{2}$. Then,  there exists an optimal pure quantum $(r=q-1-v, \delta=v+1)$-locally recoverable code with parameters
$$[[\lambda n n_2 \cdots n_w, \lambda n n_2 \cdots n_w - 2 \lambda v n_2 \cdots n_w, v+1]]_{q}.$$
\item[C] Suppose $s=2\varsigma > 2$ is a positive integer and set $n=q^2-1$.  Take positive integers $n_\ell$, $2 \leq n_\ell < q^2, 2 \leq \ell \leq w$ and $v \leq \frac{q^2-2}{q+1}$. Then,  there exists an optimal pure quantum $(r=q^2-1-v, \delta=v+1)$-locally recoverable code with parameters
$$[[\lambda n n_2 \cdots n_w, \lambda n n_2 \cdots n_w - 2 \lambda v n_2 \cdots n_w, v+1]]_{q}.$$
The above result is also true for $q$ odd, $1 \leq v \leq 2q-3$ and $\lambda =2$.
\end{description}
\end{teo}

We conclude with six tables, Tables \ref{tab:quantum-lrcs1} to \ref{tab:quantum-lrcs6}. These tables compile key parameters of optimal pure quantum $(r, \delta)$-LRCs, including the qudit dimension, locality, length, number of information qudits, maximum length, and the value of $r+\delta -1$.  Table \ref{tab:quantum-lrcs1} presents the results established in this paper, while the remaining tables summarize findings from the cited literature. The original notation from each respective reference has been preserved throughout. Note that in Table \ref{tab:quantum-lrcs3}, $\square_q$ represents the set of quadratic residues (squares) in the finite field $\mathbb{F}_q$.

\begin{landscape}
\begin{table}[p]
    \centering

    \caption{Parameters of optimal pure quantum $(r,\delta)$-LRCs extracted from this paper}
    \label{tab:quantum-lrcs1}
    \scalebox{0.75}{
    \begin{tabular}{ccccccccc}
\toprule
Ref. & \begin{tabular}[c]{@{}c@{}}Qudit\\ Dim ($q$)\end{tabular} & $r$ & $\delta$ & \begin{tabular}[c]{@{}c@{}}Code\\ Length\end{tabular} & \begin{tabular}[c]{@{}c@{}}Information\\ Qudits ($k$)\end{tabular} & \begin{tabular}[c]{@{}c@{}}Max Code\\ Length\end{tabular} & \begin{tabular}[c]{@{}c@{}}Max\\ $r+\delta-1$\end{tabular} & Parameter Constraints / Symbol Ranges \\
\midrule
Th. 30 & $q$ & $n+1-\delta$ & $3, 5, \dots$ & $\lambda n$ & $\lambda(n+2-2\delta)$ & unbounded & $q^2+1$ & $\lambda \in \mathbb{N}^+ \land n = m_1^2 + m_2^2 \land m_i \in \mathbb{N}^+ \land \delta \leq m_1 + m_2$ \\
Th. 31(1) & $q$ & $q-\delta$ & $2, 3, \dots$ & $\lambda(q-1)$ & $\lambda(q+1-2\delta)$ & unbounded & $q-1$ & $\lambda \in \mathbb{N}^+ \land \delta \le \lfloor q/2 \rfloor$ \\
Th. 31(2) & $q$ & $q^2-\delta$ & $2, 3, \dots$ & $\lambda(q^2-1)$ & $\lambda(q^2+1-2\delta)$ & unbounded & $q^2-1$ & $\lambda \in \mathbb{N}^+ \land \delta \le q-1$ \\
Cor. 32 & $q$ odd & $q^2-\delta$ & $2, 3, \dots$ & $2(q^2-1)$ & $2(q^2+1-2\delta)$ & $2(q^2-1)$ & $q^2-1$ & $\delta \le 2q-2$ \\
Th. 34 A & $q$ & $n+1-\delta$ & $3, 5, \dots$ & $\lambda n_2 \cdots n_w n$ & $\lambda n_2 \cdots n_w (n+2-2\delta)$ & unbounded & $q^2+1$ & $\lambda, n_2, \dots, n_w \in \mathbb{N}^+ \land n = m_1^2 + m_2^2 \land \delta \leq m_1 + m_2$  \\
Th. 34 B & $q$ & $q-\delta$ & $2, 3, \dots$ & $\lambda n_2 \cdots n_w(q-1)$ & $\lambda n_2 \cdots n_w(q+1-2\delta)$ & unbounded & $q-1$ & $\lambda, n_2, \dots, n_w \in \mathbb{N}^+ \land \delta \le \lfloor q/2 \rfloor$ \\
Th. 34 C & $q$ & $q^2-\delta$ & $2, 3, \dots$ & $\lambda n_2 \cdots n_w(q^2-1)$ & $\lambda n_2 \cdots n_w(q^2+1-2\delta)$ & unbounded & $q^2-1$ & $\lambda, n_2, \dots, n_w \in \mathbb{N}^+ \land \delta \le q-1$ \\
Th. 34 C &  $q$ odd & $q^2-\delta$ & $2, 3, \dots$ & $2n_2 \cdots n_w(q^2-1)$ & $2n_2 \cdots n_w(q^2+1-2\delta)$ & unbounded & $q^2-1$ & $n_2, \dots, n_w \in \mathbb{N}^+ \land \delta \le 2q-2$ \\
\bottomrule
\end{tabular}
    }

    \vspace{1cm} 

    \caption{Parameters of optimal pure quantum $(r,\delta)$-LRCs extracted from \cite{QMP-3}}
    \label{tab:quantum-lrcs2}
    \scalebox{0.75}{
    \begin{tabular}{ccccccccc}
\toprule
Ref. & \begin{tabular}[c]{@{}c@{}}Qudit\\ Dim ($q$)\end{tabular} & $r$ & $\delta$ & \begin{tabular}[c]{@{}c@{}}Code\\ Length\end{tabular} & \begin{tabular}[c]{@{}c@{}}Information\\ Qudits ($k$)\end{tabular} & \begin{tabular}[c]{@{}c@{}}Max Code\\ Length\end{tabular} & \begin{tabular}[c]{@{}c@{}}Max\\ $r+\delta-1$\end{tabular} & Parameter Constraints / Symbol Ranges \\
\midrule
Th. 5.1 & $q \ge 5$ & $q-1-u$ & $u+1$ & $t(q-1)$ & $t(q-1)-2tu-2v$ & $(q-1)^2$ & $q-1$ & $u, t \in \mathbb{N}^+ \land v \in \mathbb{N} \land v \le u \land u+v \le \lfloor\frac{q-1}{2}\rfloor - 1$ \\
Th. 5.4 & $q \ge 7$ & $q+v-s-1$ & $s+1$ & $t(2q+v-2)$ & $t(2q+v-2)-4st-2$ & $\lfloor\frac{q+2}{6}\rfloor(2q+4)$ & $2q-2$ & $s, v, t \in \mathbb{N}^+ \land v \mid (q-1) \land v \ge 6 \land s \le \lfloor\frac{v}{2}\rfloor - 2 \land vt \le q+v-s-3$ \\
Th. 5.6 & $q \ge 9$ & $q+v-s-1$ & $s+1$ & $t(2q+v-2)$ & $t(2q+v-2)-4st-4$ & $\lfloor\frac{q+3}{8}\rfloor(2q+6)$ & $2q-2$ & $s, v, t \in \mathbb{N}^+ \land l \in \mathbb{N} \land v \mid (q-1) \land v \ge 8 \land s+l \le \lfloor\frac{v}{2}\rfloor - 1 \land l \le \lfloor\frac{s}{2}\rfloor + 1$ \\
\bottomrule
\end{tabular}
    }

 \vspace{1cm}

    \caption{Parameters of optimal pure quantum $(r,\delta)$-LRCs extracted from \cite{QMP-1}}
    \label{tab:quantum-lrcs3}
    \scalebox{0.75}{
   \begin{tabular}{ccccccccc}
\toprule
Ref. & \begin{tabular}[c]{@{}c@{}}Qudit\\ Dim ($q$)\end{tabular} & $r$ & $\delta$ & \begin{tabular}[c]{@{}c@{}}Code\\ Length\end{tabular} & \begin{tabular}[c]{@{}c@{}}Information\\ Qudits ($k$)\end{tabular} & \begin{tabular}[c]{@{}c@{}}Max Code\\ Length\end{tabular} & \begin{tabular}[c]{@{}c@{}}Max\\ $r+\delta-1$\end{tabular} & Parameter Constraints / Symbol Ranges \\
\midrule
Cor. 37 & odd $q$ & $q-i$ & $i+1$ & $q^2$ & $2((q-i)(q-1)+(q-j))-q^2$ & $q^2$ & $q$ & $i, j \in \mathbb{N}^+ \land \frac{j-1}{2} \le i \le j < \frac{q}{2}$ \\
Cor. 38 & odd $q$ & $h-i$ & $i+1$ & $h^2$ & $2((h-i)(h-1)+(h-j))-h^2$ & $q^2$ & $q$ & $h, i, j \in \mathbb{N}^+ \land (h-1) \mid (q-1) \land (1-h) \in \square_q \land \frac{j-1}{2} \le i \le j < \frac{h}{2}$ \\
Th. 39 & odd $q$ & $q-t$ & $t+1$ & $q^2$ & $2(q(q-t)-(d-t-1))-q^2$ & $q^2$ & $q$ & $t, d \in \mathbb{N}^+ \land t < \frac{q}{2} \land t+1 \le d \le \min\{2(t+1), \frac{q^2-t(2q-2)+2}{2}\}$ \\
Th. 40 & odd $q$ & $h-t$ & $t+1$ & $mh$ & $2(m(h-t)-(d-t-1))-mh$ & $q^2$ & $q$ & $m, h, t, d \in \mathbb{N}^+ \land m, h \le q \land (h-1) \mid (q-1) \land t < \frac{h}{2}$ \\
Cor. 44 & odd $q$ & $q^2-a$ & $a+1$ & $q^4$ & $2((q^2-a)(q^2-1)+(q^2-b))-q^4$ & $q^4$ & $q^2$ & $a, b \in \mathbb{N} \land a, b < q-1 \land a \le b \le 2a$ \\
\bottomrule
\end{tabular}
    }
\end{table}
\end{landscape}

\begin{landscape}
\begin{table}[p]
    \centering

    \caption{Parameters of optimal pure quantum $(r,\delta)$-LRCs extracted from \cite{QMP-4}}
    \label{tab:quantum-lrcs4}
    \scalebox{0.75}{
    \begin{tabular}{ccccccccc}
\toprule
Ref. & \begin{tabular}[c]{@{}c@{}}Qudit\\ Dim ($q$)\end{tabular} & $r$ & $\delta$ & \begin{tabular}[c]{@{}c@{}}Code\\ Length\end{tabular} & \begin{tabular}[c]{@{}c@{}}Information\\ Qudits ($k$)\end{tabular} & \begin{tabular}[c]{@{}c@{}}Max Code\\ Length\end{tabular} & \begin{tabular}[c]{@{}c@{}}Max\\ $r+\delta-1$\end{tabular} & Parameter Constraints / Symbol Ranges \\
\midrule
\midrule
Th. 12 & $q$ & $k_1$ & $q-k_1+1$ & $2q$ & $2(k_1+k_2-q)$ & $2q$ & $q$ & $k_1, k_2 \in \mathbb{N}^+ \land 1 \le k_2 \le k_1 \le q-1 \land k_1+k_2 \ge q \land 2k_1-k_2 \le q+1$ \\
Th. 13 & $q > 2$ & $q^2-l_1-1$ & $l_1+2$ & $2q^2$ & $2(q^2-l_1-l_2-2)$ & $2q^2$ & $q^2$ & $l_1, l_2 \in \mathbb{N} \land 0 \le l_1 \le q-2 \land l_1 \le l_2 \le 2l_1+2$ \\
Th. 14 & $q > 2$ & $k_1$ & $q-k_1+1$ & $3q$ & $4k_1+2k_2-3q$ & $3q$ & $q$ & $k_1, k_2 \in \mathbb{N}^+ \land 1 \le k_2 \le k_1 \le q-1 \land 2k_1-k_2 \le q+1 \land (k_1+k_2 > q \lor k_1=k_2=q/2)$ \\
Th. 15 & $q > 2$ & $q^2-l_1-1$ & $l_1+2$ & $3q^2$ & $3q^2-4l_1-2l_2-6$ & $3q^2$ & $q^2$ & $l_1, l_2 \in \mathbb{N} \land 0 \le l_1 \le q-2 \land l_1 \le l_2 \le 2l_1+2$ \\
\bottomrule
\end{tabular}
    }

    \vspace{1cm} 

    \caption{Parameters of optimal pure quantum $(r,\delta)$-LRCs extracted from \cite{QMP-5}}
    \label{tab:quantum-lrcs5}
    \scalebox{0.65}{
    \begin{tabular}{ccccccccc}
\toprule
Ref. & \begin{tabular}[c]{@{}c@{}}Qudit\\ Dim ($q$)\end{tabular} & $r$ & $\delta$ & \begin{tabular}[c]{@{}c@{}}Code\\ Length\end{tabular} & \begin{tabular}[c]{@{}c@{}}Information\\ Qudits ($k$)\end{tabular} & \begin{tabular}[c]{@{}c@{}}Max Code\\ Length\end{tabular} & \begin{tabular}[c]{@{}c@{}}Max\\ $r+\delta-1$\end{tabular} & Parameter Constraints / Symbol Ranges \\
\midrule
Th. 5 & odd $q \ge 5$ & $q-1-u$ & $u+1$ & $(2m+1)(q-1)t$ & $(2m+1)(q-1)t-2(2m+1)tu-2v$ & unbounded & $q-1$ & $m, u, t \in \mathbb{N}^+ \land v \in \mathbb{N} \land 2m \mid (q^2-1) \land v \le u \land u+v \le \frac{q-3}{2}$ \\
Th. 6 & odd $q \ge 9$ & $q+v-s-1$ & $s+1$ & $(2m+1)(2q+v-2)t$ & $(2m+1)(2q+v-2)t-4(2m+1)st-2l$ & $3\lfloor\frac{q+1}{12}\rfloor(2q+2)$ & $\frac{4q-4}{3}$ & $m, s, v, t \in \mathbb{N}^+ \land l \in \mathbb{N} \land 2m \mid (q^2-1) \land v \mid (q-1) \land s+l \le \lfloor\frac{v}{2}\rfloor - 1 \land l \le \lfloor\frac{s}{2}\rfloor + 1$ \\
\bottomrule
\end{tabular}
    }

 \vspace{1cm}

    \caption{Parameters of optimal pure quantum $(r,\delta)$-LRCs extracted from \cite{QMP-2}}
    \label{tab:quantum-lrcs6}
    \scalebox{0.75}{
   \begin{tabular}{ccccccccc}
\toprule
Ref. & \begin{tabular}[c]{@{}c@{}}Qudit\\ Dim ($q$)\end{tabular} & $r$ & $\delta$ & \begin{tabular}[c]{@{}c@{}}Code\\ Length\end{tabular} & \begin{tabular}[c]{@{}c@{}}Information\\ Qudits ($k$)\end{tabular} & \begin{tabular}[c]{@{}c@{}}Max Code\\ Length\end{tabular} & \begin{tabular}[c]{@{}c@{}}Max\\ $r+\delta-1$\end{tabular} & Parameter Constraints / Symbol Ranges \\
\midrule
Th. 8 & $q \ge 5$ & $q-1-u$ & $u+1$ & $Nt(q-1)$ & $Nt(q-1)-2Ntu-2v$ & unbounded & $q-1$ & $N, u, t \in \mathbb{N}^+ \land v \in \mathbb{N} \land N \mid (q^2-1) \land N \nmid (q+1) \land v \le u \land 2u+v \le q-2$ \\
Th. 9 & $q \ge 8$ & $q+v-s-1$ & $s+1$ & $Nt(2q+v-2)$ & $Nt(2q+v-2)-4Nst-2l$ & $3\lfloor\frac{q+1}{12}\rfloor(2q+2)$ & $\frac{4q-4}{3}$ & $N, s, v, t \in \mathbb{N}^+ \land l \in \mathbb{N} \land N \mid (q^2-1) \land N \nmid (q+1) \land v \mid (q-1) \land 2s+l \le v-1$ \\
Th. 10 & odd $q \ge 5$ & $q-1-u$ & $u+1$ & $2t(q-1)$ & $2t(q-1)-4tu-2v$ & unbounded & $q-1$ & $u, t \in \mathbb{N}^+ \land v \in \mathbb{N} \land v \le u \land u+v \le \frac{q-3}{2}$ \\
Th. 12 & odd $q \ge 5$ & $m-u$ & $u+1$ & $Ntm$ & $Ntm-2Ntu-2v$ & unbounded & $q-1$ & $N, s, u, t, m \in \mathbb{N}^+ \land v \in \mathbb{N} \land q-1 = (N-1)s \land N \ge 3 \land s \ge 2 \land p \mid N \land v \le u$ \\
Th. 13 & odd $q \ge 5$ & $m-u$ & $u+1$ & $Ntm$ & $Ntm-2Ntu-2v$ & unbounded & $q-1$ & $N, s, u, t, m \in \mathbb{N}^+ \land v \in \mathbb{N} \land q-1 = (N-1)s \land N \ge 3 \land s \ge 2 \land p \nmid N \land v \le u$ \\
\bottomrule
\end{tabular}
    }
\end{table}
\end{landscape}

\section*{Data availability}
No databases were generated or analyzed during this study.

\section*{Conflict of interest}
The authors declare they have no conflict of interest.

\section*{Use of AI}
Initial proofs of Proposition \ref{sumsquare}, and Lemmas \ref{lem:selforthogonal} and \ref{lemauno}, as well as, Tables \ref{tab:quantum-lrcs2} to \ref{tab:quantum-lrcs6} were obtained with the help of AI (Google Gemini). Proofs were completely rewritten by the authors, and no AI-generated texts or equations remain in this paper.


\end{document}